\documentclass[10pt]{article}
\usepackage[a4paper, total={6.5in, 8in}]{geometry}

\usepackage[T1]{fontenc}
\usepackage{lmodern}
\usepackage{amsmath,amsthm,amscd,amssymb,latexsym,esint,upref,stmaryrd,enumerate,xcolor,verbatim,yfonts,multicol}
\usepackage{fancyhdr}
\usepackage{graphicx}

\usepackage{pifont}
\usepackage{dcolumn}
\usepackage{bm}
\usepackage{caption,subcaption}

\usepackage{amsthm}
\usepackage{latexsym}
\usepackage{graphics}
\usepackage{amssymb}
\usepackage{amsfonts}
\usepackage{mathrsfs}
\usepackage{bm}
\usepackage{amsmath}
\usepackage{color}
\usepackage[dvipsnames]{xcolor}
\usepackage{arydshln} 
\usepackage[title]{appendix}

\usepackage{authblk}

\usepackage{cancel}

\RequirePackage{graphicx}
\RequirePackage{flushend}
\RequirePackage[colorlinks,citecolor=blue,urlcolor=blue,linkcolor=blue]{hyperref}

\usepackage{mathtools}  

\newtheorem{theorem}{Theorem}[section]
\newtheorem{assumption}[theorem]{Assumption}

\newtheorem{definition}[theorem]{Definition}
\newtheorem{hypothesis}[theorem]{Hypothesis}
\theoremstyle{remark}
\newtheorem{remark}[theorem]{Remark}
\newcommand{\oh}{o}

\newcommand{\R}{{\mathbb R}}
\newcommand{\N}{{\mathbb N}}

\newcommand{\C}{{\mathbb C}}

\newcommand{\dott}{\,\cdot\,}
\newcommand{\hatt}{\widehat} 
\newcommand{\wti}{\widetilde}

\renewcommand{\a}{\alpha}
\renewcommand{\b}{\beta}

\newcommand{\z}{\zeta}

\DeclareMathOperator{\rank}{rank}

\DeclareMathOperator{\Ai}{Ai}

\begin{document}

\title{The spectral $\zeta$-function of Sturm--Liouville operators with $N$ generalized point potentials}

\bigskip

\author[1]{Guglielmo Fucci\footnote{e-mail: fuccig@ecu.edu}}
\affil[1]{\small{Department of Mathematics, East Carolina University}\\ \small{Greenville, North Carolina 27858, USA.}}

\author[2]{Jonathan Stanfill}
\affil[2]{Division of Geodetic
Science, School of Earth Science, The Ohio State University, Columbus, OH 43210, USA}

\bigskip
\bigskip

\maketitle

\begin{abstract}

This work analyzes the spectral zeta function associated with regular and singular Sturm--Liouville operators endowed with $N$ generalized point potentials. The point interactions are characterized as a particular class of self-adjoint extensions of Sturm--Liouville operators defined on a chain of $N+1$ adjacent intervals. Each interaction is described by $SL(2,\R)$ matching conditions relating generalized boundary values across neighboring intervals. 
We construct the spectral zeta function associated with Sturm--Liouville operators endowed with $N$ generalized point potentials in terms of a contour integral involving an appropriate characteristic function. The process of analytic continuation of the spectral zeta function to a neighborhood of the origin is illustrated by means of specific examples where the $\zeta$-regularized functional determinant is also explicitly computed.

\end{abstract}
\section{Introduction}\label{sec:I}

Exactly solvable systems play an important role in mathematical  physics since they provide a model of the salient characteristics 
of a physical system and, at the same time, have the advantage of being described by equations that possess closed form solutions. The classic point potentials represent a particular example of exactly solvable models \cite{albeverio} and they are mainly utilized to describe idealized boundaries, localized interactions, and singular perturbations in quantum mechanics and quantum field theory \cite{belloni14}. 

The most widely studied point interactions for Schr\"odinger operators are the $\delta$ potential, the $\delta'$ potential, and a combination of the two, namely, the $\delta$-$\delta'$ potential. 
The $\delta$ potential was the first to be extensively considered especially for its role in the modeling of crystals \cite{kronig-prsa31}. Its applications extend to many other areas such as nuclear physics, molecular and solid state physics, just to name a few. The $\delta'$ potential was originally considered in \cite{SE} and later on in the ambit of non-relativistic quantum mechanics in \cite{albeverio,antoine87} (see also \cite{ADK98,G93} for different characterizations of $\delta'$). Applications of this potential to higher dimensional quantum field theory have been presented, e.g., in \cite{cavero2021casimir,nieto2017towards}. 
The $\delta$-$\delta'$ potential recently became an important topic of research and was introduced mainly as a generalization of the $\delta$ and the $\delta'$ potentials. Its main formal properties were investigated in \cite{golovaty} while its applications to a variety of physical systems can be found, for instance, in \cite{cavero2021casimir,munoz2015delta,munoz2019hyperspherical}. The amount of literature devoted to these point potentials in one and higher dimensions is quite large and the interested reader can find an in-depth treatment of the subject in the monograph \cite{albeverio}.   

A $\delta$ potential can be defined as a limit of scaled local potentials (see e.g. \cite{albeverio81,albeverio84} and references therein) as it is often 
introduced in quantum mechanics textbooks. However, this approach is less than desirable because it does not offer a straightforward mathematical framework for considering more generalized point potentials. A more suitable alternative that allows for the definition and analysis of generalized point potentials consists in utilizing the theory of self-adjoint extensions of Sturm--Liouville operators and Schr\"{o}dinger operators, see, in particular \cite{albeverio,ADK98,albeverio2000singular,G93,SE,seba86} and the upcoming \cite{FRS26}.

In this work we follow the latter viewpoint and define point potentials as realizations of self-adjoint extensions of Sturm--Liouville operators, providing a straightforward generalization from the Schr\"odinger setting. More precisely, let $R$ be an interior point where the potential is located. We regard $R$ as a common boundary (point) of {\it two adjacent} intervals. These intervals could be both semi-infinite (the case that is usually analyzed in the literature \cite{albeverio}) or could be both finite, or one of each. In this framework, the potential at $R$ is defined in terms of the boundary conditions that characterize a specific class of self-adjoint extensions of the Sturm--Liouville operator on the two adjoining intervals. Although the family of self-adjoint extensions on two intervals is quite large, in the next section we identify precisely those extensions that describe point interactions. This local characterization serves as the building block for the general framework developed in this paper, where a chain of $N+1$ adjacent intervals is coupled by $N$ generalized point potentials.

The generalized point potentials constructed in this fashion expand considerably the class of exactly solvable models in quantum mechanics and quantum field theory. 
We would like to point out that the boundary of a region on which the fields satisfy self-adjoint boundary conditions (of which Dirichlet, Neumann and Robin are particular cases analyzed often in the literature) could also be viewed, following the description provided above, as a point potential.
From this perspective, a point potential becomes a {\it semi-transparent} boundary which can be used, for instance, to model the interaction of fields with real materials.  

Most investigations of point interactions focus on a single interaction on the real line or on a finite interval. Although the corresponding self-adjoint extension theory is well understood, considerably less is known about systems consisting of arbitrary Sturm--Liouville operators with several generalized point interactions. Such systems naturally arise when modeling multiple interfaces, layered media, quantum graphs, and semitransparent boundaries.

The main goal of this work is to construct the spectral zeta function associated with Sturm--Liouville operators endowed with generalized point potentials. To this end, we first describe a theory of $N$ generalized point potentials based on self-adjoint extensions of Sturm--Liouville operators defined on a chain of $N+1$ adjacent intervals. Each interaction is described by an arbitrary $SL(2,\mathbb{R})$ matching matrix acting on generalized boundary values. This framework applies to both regular and singular Sturm--Liouville problems, accommodates any combination of limit point and limit circle endpoints, allows for different operators in each interval, and contains the classical $\delta$, $\delta'$, and $\delta$--$\delta'$ interactions as special cases. It therefore provides a unified operator-theoretic description of a broad class of exactly solvable models with multiple point interactions.

After establishing this operator-theoretic framework, we  construct the spectral zeta function associated with Sturm--Liouville operators endowed with $N$ generalized point potentials. To this end, we extend the methods developed in  \cite{FGKS21,FPS25,FPS26,gesztesy19} for the construction of the spectral zeta function of Sturm--Liouville operators on a single interval to the setting of a chain of $N+1$ adjacent intervals coupled through generalized point interactions. The key ingredient is the construction of suitable characteristic functions for each class of self-adjoint extensions describing the point potentials. These characteristic functions naturally incorporate the transfer matrices associated with the individual interactions in a way that encodes the cumulative effect of all point potentials. The characteristic function is then used as the starting point of the contour integral representations of the corresponding spectral zeta functions. By analyzing the large spectral parameter asymptotics of the characteristic functions, we obtain the analytic continuation of the spectral zeta function to a neighborhood of the origin. As an application, we compute the functional determinant of self-adjoint realizations $T$, or equivalently $\zeta'(0;T)$, for the interface and the two generalized point potential system.

To the best of our knowledge, this analysis is new for a few reasons. First, the existing literature on point interactions is largely devoted to Schr\"odinger operators on the real line or on a single finite interval, whereas our framework applies equally well to Sturm--Liouville problems on bounded and unbounded intervals with any combination of limit point and limit circle endpoints as well as for different operators on each interval. Second, previous investigations of spectral zeta functions for point interactions have primarily been motivated by applications to quantum field theory, such as the Casimir effect, and have been restricted to a few explicitly solvable interactions, most notably the $\delta$ and $\delta$-$\delta'$ potentials for Schr\"odinger operators (see \cite{cavero2021casimir,munoz2015delta,munoz2019hyperspherical,romaniega2022casimir} and the references therein). The present work, instead, develops the spectral zeta function for the considerably broader class of point potentials arising from general self-adjoint extensions of Sturm--Liouville operators.

The outline of the paper is as follows. In Section \ref{Sec2} we develop the self-adjoint extension framework and characterize the class of generalized point potentials generated by interior
boundary conditions. Section \ref{Sec3} is devoted to the construction of the characteristic functions and the contour integral representation of the spectral zeta function.
Finally, we analyze the analytic continuation procedure and illustrate the general formalism through explicit examples.

\section{Self-adjoint extensions and the point potential}\label{Sec2}

To set the stage, for $j\in {\cal S}=\{1,\ldots,N+1\}$, we consider the intervals $I_j=(a_j,b_j)$ with $-\infty\leq a_j<b_j\leq +\infty$, and the following Sturm--Liouville (SL) differential expression
\begin{equation}\label{2}
\tau_j=\frac{1}{r_j(x)}\left[-\frac{d}{dx}p_j(x)\frac{d}{dx}+q_j(x)\right],
\end{equation}
where the coefficient functions satisfy the following:
\begin{hypothesis}\label{Hyp1}
The Lebesgue measurable functions $p_j,q_j,r_j$ on $I_j$ are such that
\\[1mm] 
$(i)$ \hspace*{1.1mm} $r_j>0$ a.e.~on $I_j$, $r_j\in L_{\textrm{loc}}(I_j,dx)$; \\[1mm] 
$(ii)$ \hspace*{.1mm} $p_j>0$ a.e.~on $I_j$, $1/p_j \in L_{\textrm{loc}}(I_j,dx)$; \\[1mm] 
$(iii)$ $q_j$ is real-valued a.e.~on $I_j$, $q_j\in L_{\textrm{loc}}(I_j,dx)$. 
\end{hypothesis}

\begin{remark}
$(i)$ In physical applications the differential expression in \eqref{2} is the same in each interval.\\[1mm]
$(ii)$ In general, the intervals $I_j$ can be either joint at the edges or disjoint.  
\end{remark}

To keep our analysis general, we assume that the endpoints of $I_j$ can be either limit circle or limit point: If for every $z\in\C$ all solutions of $\tau_j u=\lambda u$ are in $L^{2}(I_j,r_jdx)$ near $a_j$ we say that the endpoint $a_j$ is limit circle. If, on the other hand, for every $z\in\C/\R$ there exists only one solution of $\tau_j u=\lambda u$ that is in $L^{2}(I_j,r_jdx)$ near $a_j$, then we say the endpoint $a_j$ is limit point. Similar definitions hold for the endpoint $b_j$. According to Weyl's alternative, these are the only two cases that can occur. It is important to point out that regular endpoints are a particular case of limit circle endpoints. Moreover, the limit point or circle classification of an endpoint depends entirely on the coefficient functions of the differential expression \eqref{2}.

From the differential expression \eqref{2} we can define the maximal operator 
\begin{align}
   & T_{\max}^{(j)}f=\tau_j f\nonumber\\
    &f_j\in D_{\textrm{max}}^{(j)}=\{f\in L^{2}(I_j,r_jdx)|f,p_jf'\in AC_{\textrm{loc}}(I_j)\;,\tau_j f\in L^{2}(I_j,r_jdx)\},
\end{align}
and the minimal operator
\begin{align}
 &T_{\min}^{(j)}f=\tau_jf \nonumber\\
 &f\in D_{\textrm{min}}^{(j)}=\{f\in D_{\textrm{max}}^{(j)}|\,W_j(h,f)(a_j)=0=W_j(h,f)(b_j),\;\textrm{for all}\;h\in D_{\textrm{max}}^{(j)}\},
\end{align}
where the Wronskian $W_j(f,g)(x)$, of $f,g\in AC_{\textrm{loc}}(I_j)$ is defined as
\begin{align}\label{6a}
   W_j(f,g)(x)=f(x)p_{j}(x)g'(x)-g(x)p_{j}(x)f'(x),\quad x\in I_{j}. 
\end{align}
Most systems in the ambit of quantum mechanics and quantum field theory are described by differential operators whose spectrum is bounded from below. For this reason we will henceforth assume that
\begin{hypothesis}\label{Hyp2}
    For $j\in{\cal S}$, the Sturm--Liouville differential expression \eqref{2} satisfies Hypothesis \ref{Hyp1} and its minimal operator $T_{\min}^{(j)}$ is bounded from below, that is
    \begin{align}\label{boundedbel}
        (u,[T_{\min}^{(j)}-\lambda_{0}^{(j)} I]u)_{L^{2}(I_j,r_jdx)}\geq 0,\quad u\in D_{\textrm{min}}^{(j)}.
    \end{align}
\end{hypothesis}
If $T_{\min}^{(j)}$ is bounded from below then, for fixed $c_j, d_j \in (a_j,b_j)$, $c_j \leq d_j$, there exists a $\nu_0^{(j)}\in\R$ such that for all $\lambda_{j}<\nu_0^{(j)}$, $\tau_j u = \lambda_j u$ has $($real-valued\,$)$ nonvanishing solutions
$u_{a_j}(\lambda,\dott) \neq 0$,
$\hatt u_{a_j}(\lambda,\dott) \neq 0$ in $(a_j,c_j]$, and $($real-valued\,$)$ nonvanishing solutions
$u_{b_j}(\lambda,\dott) \neq 0$, $\hatt u_{b_j}(\lambda,\dott) \neq 0$ in $[d_j,b_j)$, such that 
\begin{align}
&W_j(\hatt u_{a_j} (\lambda,\dott),u_{a_j}(\lambda,\dott)) = 1,
\quad u_{a_j}(\lambda,x)=\oh(\hatt u_{a_j}(\lambda,x))
\text{ as $x\downarrow a_j$,} \label{A3} \\
&W_j(\hatt u_{b_j} (\lambda,\dott),u_{b_j}(\lambda,\dott))\, = 1,
\quad u_{b_j}(\lambda,x)\,=\oh(\hatt u_{b_j}(\lambda,x))
\text{ as $x\uparrow b_j$,} \\
&\int_{a_j}^{c_j} dx \, p_j(x)^{-1}u_{a_j}(\lambda,x)^{-2}=\int_{d_j}^{b_j} dx \, 
p_j(x)^{-1}u_{b_j}(\lambda,x)^{-2}=\infty,\label{A4} \\
&\int_{a_j}^{c_j} dx \, p_j(x)^{-1}{\hatt u_{a_j}(\lambda,x)}^{-2}<\infty, \quad 
\int_{d_j}^{b_j} dx \, p_j(x)^{-1}{\hatt u_{b_j}(\lambda,x)}^{-2}<\infty, \label{A5}
\end{align}
a proof of which can be found, for instance, in \cite{LM36} (see also \cite{Re43,Re51} and \cite[Appendix]{HW55}).
\begin{definition}
Suppose that $T_{\min}^{(j)}$ is bounded from below, and let $\lambda\in\R$. Then $u_{a_j}(\lambda,\dott)$ $($resp., $u_{b_j}(\lambda,\dott)$$)$ described above is called a {\it principal} $($or {\it minimal}\,$)$
solution of $\tau_j u=\lambda u$ at $a_j$ $($resp., $b_j$$)$. A real-valued solution 
$v_{a_j}(\lambda,\dott)$ $($resp., $v_{b_j}(\lambda,\dott)$$)$ of $\tau_j
u=\lambda u$ linearly independent of $u_{a_j}(\lambda,\dott)$ $($resp.,
$u_{b_j}(\lambda,\dott)$$)$ is called {\it nonprincipal} at $a_j$ $($resp., $b_j$$)$.
\end{definition}

Principal and nonprincipal solutions can be utilized to define the {\it generalized boundary values} at a limit circle endpoint according to the following:
\begin{theorem}[{\cite[Thm. 4.5]{GLN20}}] \label{At3}
Assume that $c_j$ is a limit circle point for $\tau_j$ and that $T_{\min}^{(j)} \geq \lambda_0 I$ for some $\lambda_0 \in \R$. Denote by 
$u_{c_j}(\lambda_0, \dott)$ and $\hatt u_{c_j}(\lambda_0, \dott)$  principal and nonprincipal solutions of $\tau_j u = \lambda_0 u$ at $c_j$, normalized so that
\begin{equation}
W_j(\hatt u_{c_j}(\lambda_0,\dott), u_{c_j}(\lambda_0,\dott)) = 1.
\end{equation}  
Then, for all $g \in D_{\max}^{(j)}$, one obtains 
\begin{align}
 \wti g(c_j) & =  - W_j(u_{c_j}(\lambda_0,\dott), g)(c_j)= \lim_{x \to c_j} \frac{g(x)}{\hatt u_{c_j}(\lambda_0,x)},\label{A6}  \\    
{\wti g}^{\, \prime}(c_j) &= W_j(\hatt u_{c_j}(\lambda_0,\dott), g)(c_j) = \lim_{x \to c_j} \frac{g(x) - \wti g(c_j) \hatt u_{c_j}(\lambda_0,x)}{u_{c_j}(\lambda_0,x)}.\label{A7} 
\end{align}
In particular, the limits on the right-hand sides of \eqref{A6} and \eqref{A7} exist. 
\end{theorem}

\begin{remark}
    It is worthwhile to mention that when an endpoint $c_j$ is regular, the principal and nonprincipal solutions can be chosen such that the generalized boundary conditions \eqref{A6} and \eqref{A7} reduce to the ordinary boundary values $g(c_j)$ and $(p_jg')(c_j)$, respectively \cite[Rem. 5.5.9(ii)]{GNZ23}.
\end{remark}

The construction of the self-adjoint extensions of a Sturm--Liouville operator generated by differential expressions of the form \eqref{2} over a countable number of intervals has been developed in \cite{Ev86} and \cite{Ev92}. For the sake of completeness, we outline here the most important aspects of that analysis. 

The maximal and minimal domains for the set of differential expressions $\{\tau_j\}_{j\in{\cal S}}$ in (\ref{2}) on the intervals $\{I_j\}_{j\in{\cal S}}$ are subsets of the Hilbert space $\mathcal{H}=\bigoplus_{j\in\cal S} L^{2}(I_j,r_jdx)$ and are written as the direct sum of the corresponding one interval domains \cite{Ev92}
\begin{equation}\label{4}
D_{\textrm{max}}=\bigoplus_{j=1}^{N+1}D_{\textrm{max}}^{(j)}\quad\textrm{and}\quad D_{\textrm{min}}=\bigoplus_{j=1}^{N+1}D_{\textrm{min}}^{(j)}.
\end{equation}
The elements of $D_{\textrm{max}}$ are the functions ${\bf f}=\{f_1,\ldots,f_{N+1}\}$ with $f_j\in D_{\textrm{max}}^{j}$. Obviously, a similar definition
holds for the elements of the minimal domain.  
The inner product of ${\bf f},{\bf g}\in \mathcal{H}$ is defined as $({\bf f},{\bf g})_\mathcal{H}=\sum_{j\in\cal S}(f_j,g_j)_{L^{2}(I_j,r_jdx)}$ while the Wronskian in \eqref{6a} can be extended to functions in $D_{\max}$ as follows
\begin{align}
    W({\bf f},{\bf g})=\sum_{j=1}^{N+1}W_j(f_j,g_j).
\end{align}
From here we can define the Lagrange sesquilinear form, which will be useful later in the analysis of the self-adjoint extensions:
\begin{align}
    [{\bf f},{\bf g}]=\sum_{j=1}^{N+1}\left[W_j(f_j,\bar{g}_j)(b_j)-W_j(f_j,\bar{g}_j)(a_j)\right].
\end{align}

The minimal, $T_{\min}$, and maximal, $T_{\max}$, operators associated with the set $\{\tau_j\}_{j\in{\cal S}}$ on the intervals $\{I_j\}_{j\in{\cal S}}$  (\ref{2}) are, then
\begin{equation}\label{5}
T_{\min/\max}=\bigoplus_{j=1}^{N+1}T_{\min/\max}^{(j)}\;,\quad\textrm{with}\quad
T_{\min/\max}^{(j)}f=\tau_j f\;,\,\,f\in D_{\textrm{min}/\textrm{max}}^{(j)}.
\end{equation}
We denote by $\tau$ the associated differential expression defined to be equal to $\tau_j$ when restricted to each interval $I_j$.
\begin{remark}\label{rem2.6}
    From \eqref{boundedbel}, it immediately follows that $T_{\min}$ is also bounded from below with
    \begin{align}
        (u,[T_{\min}-\lambda_0 I]u)_{\mathcal{H}}\geq 0,\quad u\in D_{\textrm{min}},\quad \text{where }\ \lambda_{0}=\textrm{min}\{\lambda_{0}^{(1)},\ldots,\lambda_{0}^{(N)}\}.
    \end{align}
\end{remark}
One can prove \cite[Thm. 2.1]{Ev92} that the operator $T_{\min}$ is closed and symmetric and that the deficiency indices $d^{+}$ and $d^{-}$ of $T_{\min}$ are $d^{+/-}=\sum_{j\in\cal S}d_{j}^{+/-}$ where $d_{j}^{+/-}$ are the deficiency indices of $T_{\min}^{(j)}$ \cite[Cor. 2.4]{Ev92}. Since the differential expression \eqref{2} has real coefficients, the deficiency indices are equal $d_{j}^{+}=d_{j}^{-}$ for $j\in\cal S$ \cite[Ch. XIII, Cor. 14]{Dunford63} and   
\begin{align}
   0\leq d^{+}=d^{-}=d\leq 2(N+1), 
\end{align} so that the self-adjoint extensions of $T_{\min}$ exist. In addition, $d$ simply counts the number of endpoints which are limit circle. In order to describe the self-adjoint extensions of $T_{\min}$ we need to introduce the notion of a {\it generalized boundary condition set} for the pair $\{T_{\max},T_{\min}\}$
\begin{definition}\label{def2.2}
    The elements $\{{\bf h}_i\}\in{\cal H}$, $i\in\{1,\ldots,d\}$ are a generalized boundary condition set for the pair $\{T_{\max},T_{\min}\}$ if the following hold:
    \\[1mm] 
$(i)$ \hspace*{1.1mm} ${\bf h}_i\in D_{\max}$ for $i\in\{1,\ldots,d\}$. \\[1mm] 
$(ii)$ \hspace*{.1mm} The set $\{{\bf h}_i\}$ is linearly independent in $D_{\max}$ modulo $D_{\min}$, that is $$\sum_{k=1 }^{d}c_k{\bf h}_{k}\in D_{\min}\; \Longrightarrow \;c_1=\cdots=c_k=0.$$ \\[1mm] 
$(iii)$ The set $\{{\bf h}_i\}$ satisfies $[{\bf h}_i,{\bf h}_k]=0$, for $i,k\in\{1,\ldots,d\}$. 
\end{definition}

The self-adjoint extensions of $T_{\min}$ are then provided by the following:
\begin{theorem}[{\cite[Thm. 3.1]{Ev92}}]\label{Thm2.7}
If $\{{\bf h}_j\}$ with $j\in\{1,\ldots,d\}$ is a generalized boundary conditions set for the pair $\{T_{\max},T_{\min}\}$, then all the self-adjoint extensions, $T$, of $T_{\min}$ are characterized by
\begin{align}\label{21}
   & T\,{\bf f}=\tau{\bf f},\nonumber\\
    &{\bf f}\in D(T)=\{{\bf f}\in D_{\max}|\,[{\bf f,{\bf h}}_j]=0,\,j=1,\ldots,d\}.
\end{align}
\end{theorem}

Since we have assumed that the minimal operators $T_{{\min},j}$ are bounded from below, we can express the self-adjoint extensions in terms of the generalized boundary conditions in Theorem \ref{At3}. 
\begin{theorem}\label{Thm2.8}
    Assume Hypotheses \ref{Hyp1} and \ref{Hyp2} hold. For $j\in{\cal S}$ and $\lambda_0\in\R$, let $u_{c_j}(\lambda_0,\dott)$ and $\hatt{u}_{c_j}(\lambda_0,\dott)$ be principal and nonprincipal solutions, respectively, at the endpoint $c_j$ normalized so that 
    \begin{align}\label{22}
        W_{j}(\hatt{u}_{c_j}(\lambda_0,\dott),u_{c_j}(\lambda_0,\dott))=1.
    \end{align} 
    Then $T_{[{\cal A}]}$ is a self-adjoint extension of $T_{\min}$ if and only if
    \begin{align}\label{23}
         & T_{[{\cal A}]}\,{\bf f}=\tau{\bf f},\nonumber\\
    &{\bf f}\in D\left(T_{[{\cal A}]}\right)=\left\{{\bf f}\in D_{\max}\Bigg|\,\sum_{j=1}^{N+1}\sum_{c_j\in\{a_j,b_j\}}{\cal A}(c_j)\begin{pmatrix}
	     \wti{f}_{j}(c_j) \\
         \wti{f}'_{j}(c_j)
	 \end{pmatrix}    
    =0\right\},
    \end{align}
where $[{\cal A}]=({\cal A}(a_1),{\cal A}(b_1),\ldots,{\cal A}(b_{N+1}))$ and ${\cal A}(c_j)$ are $d\times 2$ matrices such that
\begin{align}\label{23a}
    \rank\left(\begin{pmatrix}
      {\cal A}(a_1)& {\cal A}(b_1)&\cdots&  {\cal A}(a_{N+1})&  {\cal A}(b_{N+1}) 
    \end{pmatrix}\right)=d,
\end{align}
\begin{align}\label{23b}
\sum_{j=1}^{N+1}\sum_{c_j\in\{a_j,b_j\}}\sigma(c_j){\cal A}(c_j)E{\cal A}^{\ast}(c_j)=0,\quad\textrm{with}\quad E=\begin{pmatrix}
        0&-1\\
        1&0
    \end{pmatrix},    
\end{align}
and 
\begin{align}
   \sigma(c_j)=\begin{cases}
        1\;\textrm{if}\; c_j\; \textrm{is the right limit circle endpoint of $I_j$}, \\
        -1 \;\textrm{if}\; c_j\; \textrm{is the left limit circle endpoint of $I_j$},\\
        0 \;\textrm{if}\; c_j\; \textrm{is a limit point of $I_j$}.
   \end{cases}
\end{align}
\end{theorem}
\begin{proof}
    According to Theorem \ref{Thm2.7} the self-adjoint extensions of $T_{\min}$ are characterized by the conditions
    \begin{align}\label{24a}
        [{\bf f},{\bf h}_{i}]=\sum_{j=1}^{N+1}\left[W_{j}(f_{j},\bar{h}_{i,j})(b_j)-W_{j}(f_{j},\bar{h}_{i,j})(a_j)\right]=0,\quad i\in\{1,\ldots,d\}.
    \end{align}
The functions $f_{j}$ and $h_{i,j}$ belong to $D^{(j)}_{\max}$ and therefore $W_{j}(f_{j},\bar{h}_{i,j})(c_j)=0$ identically at a limit point endpoint $c_j$. For this reason, we can write
   \begin{align}\label{24}
        [{\bf f},{\bf h}_{i}]=\sum_{j=1}^{N+1}\sum_{c_j\in\{a_j,b_j\}}\sigma(c_j)W_{j}(f_{j},\bar{h}_{i,j})(c_j)=0,\quad i\in\{1,\ldots,d\}.
\end{align} 

Since the principal and nonprincipal solutions $u_{c_j}(\lambda_0,\dott)$ and $\hatt{u}_{c_j}(\lambda_0,\dott)$ satisfy the normalization conditions \eqref{22} we can write 
    \begin{align}
     \bar{h}_{i,j}(x)=W_{j}(\hatt{u}_{c_j},\bar{h}_{i,j})(c_j)u_{c_j}(\lambda_0,x)+ W_{j}(\bar{h}_{i,j},u_{c_j})(c_j) \hatt{u}_{c_j}(\lambda_0,x).
    \end{align}
 This expression allows us to compute the Wronskian
 \begin{align}
  W_{j}(f_{j},\bar{h}_{i,j})(c_j)&= W_{j}(\hatt{u}_{c_j},\bar{h}_{i,j})(c_j)W_{j}(f_j,u_{c_j})(c_j)+W_{j}(\bar{h}_{i,j},u_{c_j})(c_j)W_{j}(f_j,\hatt{u}_{c_j})(c_j)\nonumber\\
  &=\alpha_{ij}(c_j)\wti{f}_{j}(c_j)+\beta_{ij}(c_j)\wti{f}'_{j}(c_j),
 \end{align}
 where in the last line we have used \eqref{A6} and \eqref{A7} of Theorem \ref{At3} and the definitions
 \begin{align}\label{29}
  \alpha_{ij}(c_j)=W_{j}(\hatt{u}_{c_j},\bar{h}_{i,j})(c_j),\quad \beta_{ij}(c_j)=W_{j}(u_{c_j},\bar{h}_{i,j})(c_j).   
 \end{align}
 The relation in \eqref{24} can then be written as the condition in \eqref{23} in terms of the $d\times 2$ matrix
 \begin{align}\label{29a}
     {\cal A}(c_j)=\sigma(c_j)\begin{pmatrix}
	 \alpha_{ij}(c_j) & \beta_{ij}(c_j) \end{pmatrix},\quad i\in\{1,\ldots,d\}.
 \end{align}

Now, let us define a function ${\bf g}$ as a linear combination of the elements of the boundary conditions set $\{{\bf h}_i\}\in D_{\max}$ with $i\in\{1,\ldots,d\}$, that is ${\bf g}=\sum_{i=1}^{d}\omega_i{\bf h}_i$ with constants $\omega_i\in\C$.
If $\omega_i$ can be chosen such that ${\bf g}\in D_{\min}$, then we have that for any ${\bf p}\in D_{\max}$,
\begin{align}\label{30}
    W_j(p_j,g_j)(c_j)=0,\quad j\in{\cal S}. 
\end{align}
Note that this relation is identically satisfied if $c_j$ is a limit point, so it is convenient to recast it as
  \begin{align}\label{30a}
    \sigma(c_j)W_j(p_j,g_j)(c_j)=0,\quad j\in{\cal S}. 
\end{align}  
In particular, by choosing the functions $\hatt{\bf U}=(\hatt U_1,\ldots,\hatt U_{N+1})$ and ${\bf U}=(U_1,\ldots,U_{N+1})$ in $D_{\max}$ defined as (cf. \cite[Thm. 4.5]{GLN20})
\begin{align}\label{35}
  &\hatt U_{j}(x)= \hatt u_{c_j}(\lambda_0,x),\quad  U_{j}(x)=u_{c_j}(\lambda_0,x),\quad\textrm{near a limit circle endpoint $c_j$},\nonumber\\
  &\hatt U_{j}(x),\; U_{j}(x)\in D_{\max}^{(j)} ,\quad\textrm{near a limit point endpoint $c_j$},
\end{align}
the relations in \eqref{30a} can be written as $4(N+1)$ conditions  
\begin{align}\label{32}
\sigma(c_j) W_j(\hatt u_{c_j},g_j)(c_j)=0,\;\;  \sigma(c_j)W_j( u_{c_j},g_j)(c_j)=0,\quad j\in{\cal S},\; c_{j}=\{a_j,b_j\}.   
\end{align}
By substituting the linear combination $g_j=\sum_{i=1}^{d}\omega_i h_{i,j}$ in the equations \eqref{32} and by using the definitions \eqref{29} we obtain
\begin{align}
    \sum_{i=1}^{d}\omega_i \sigma(c_j)\bar{\alpha}_{ij}(c_j)=0,\;\;\sum_{i=1}^{d}\omega_i \sigma(c_j)\bar{\beta}_{ij}(c_j)=0,\quad j\in{\cal S},\; c_{j}=\{a_j,b_j\},
\end{align}
which can be more concisely written in matrix form as
\begin{align}\label{36}
   M\begin{pmatrix}
        \omega_1\\
        \omega_2\\
        \vdots\\
        \omega_d
    \end{pmatrix}= \begin{pmatrix}
        {\cal A}^{\ast}(a_1)\\
        {\cal A}^{\ast}(b_1)\\
        \vdots\\
        {\cal A}^{\ast}(a_{N+1})\\
        {\cal A}^{\ast}(b_{N+1})\\
    \end{pmatrix}\begin{pmatrix}
        \omega_1\\
        \omega_2\\
        \vdots\\
        \omega_d
    \end{pmatrix}=0,
\end{align}
where $M$ is a $4(N+1)\times d$ matrix and the matrix ${\cal A}^{\ast}(c_j)$ denotes the conjugate transpose of ${\cal A}(c_j)$. 
Since the elements of the boundary conditions set $\{{\bf h}_i\}$ with $i\in\{1,\ldots,d\}$ are linearly independent modulo $D_{\min}$, the only solution of the system \eqref{36} must be the trivial one. This occurs if and only if $\rank(M)=d$. Since the rank is invariant under complex conjugation and transposition, we can conclude that 
\begin{align}
    \rank\left(\begin{pmatrix}
      {\cal A}(a_1)& {\cal A}(b_1)&\cdots&  {\cal A}(a_{N+1})&  {\cal A}(b_{N+1}) 
    \end{pmatrix}\right)=d.
\end{align}

The generalized boundary condition set $\{{\bf h_i}\}$ satisfies condition $(iii)$ of Definition \ref{def2.2}, that is
\begin{align}
     [{\bf h}_i,{\bf h}_{k}]=\sum_{j=1}^{N+1}\sum_{c_j\in\{a_j,b_j\}}\sigma(c_j)W_{j}(h_{i,j},\bar{h}_{k,j})(c_j)=0,\quad i,k\in\{1,\ldots,d\}.
\end{align}
By utilizing the Pl\"ucker relation for Wronskians\footnote{Assuming that $f_j\in AC_{\textrm{loc}}(a, b)$
$j = \{1, 2, 3, 4\}$, the following elementary algebraic identity (Pl\"ucker relation) holds:
$$W(f_1,f_2)(x)W(f_3,f_4)(x)+W(f_1,f_3)(x)W(f_4,f_2)(x)+W(f_1,f_4)(x)W(f_2,f_3)(x)=0,\;\textrm{for a.e.}\;x\in(a,b).$$} 
\begin{align}\label{Eq:Plucker}
  W_{j}(h_{i,j},\bar{h}_{k,j})(c_j) W_{j}(\hatt{u}_{c_j},u_{c_j})(c_j)&+ W_{j}(h_{i,j},\hatt{u}_{c_j})(c_j) W_{j}(u_{c_j},\bar{h}_{k,j})(c_j)\notag\\
  &+W_{j}(h_{i,j},u_{c_j})(c_j) W_{j}(\bar{h}_{k,j},\hatt{u}_{c_j})(c_j)=0,
\end{align}
and the definition \eqref{29} we obtain
\begin{align}
[{\bf h}_i,{\bf h}_{k}]=\sum_{j=1}^{N+1}\sum_{c_j\in\{a_j,b_j\}}\sigma(c_j)\left(\bar{\alpha}_{ij}(c_j)\beta_{kj}(c_j)-\bar{\beta}_{ij}(c_j)\alpha_{kj}(c_j)\right)=0,\quad i,k\in\{1,\ldots,d\}.
\end{align}
By noting that at a limit circle endpoint $c_j$ we have 
\begin{align}
    {\cal A}(c_j)EA^{\ast}(c_j)=\bar{\alpha}_{ij}(c_j)\beta_{kj}(c_j)-\bar{\beta}_{ij}(c_j)\alpha_{kj}(c_j),\quad\textrm{with}\quad E=\begin{pmatrix}
        0&-1\\
        1&0
    \end{pmatrix},
\end{align}
we obtain the condition
\begin{align}
\sum_{j=1}^{N+1}\sum_{c_j\in\{a_j,b_j\}}\sigma(c_j){\cal A}(c_j)EA^{\ast}(c_j)=0.    
\end{align}

Now, let us assume that the operator $T_{[{\cal A}]}$ is defined as in \eqref{23} with matrices ${\cal A}(c_j)$ satisfying the conditions \eqref{23a} and \eqref{23b}. For ${\bf f}\in D_{\min}$ we have that $W_{j}(h_j,f_j)(a_j)=W_j(h_j,f_j)(b_j)=0$ for all ${\bf h}\in D_{\max}$ with $j\in{\cal S}$. In particular, by choosing $\hatt{\bf U},{\bf U}\in D_{\max}$ as in \eqref{35} one finds
$W_{j}(u_{c_j},f_j)(c_j)=W_{j}(\hatt{u}_{c_j},f_j)(c_j)=0$ which implies that ${\bf f}$ satisfies the condition in \eqref{23} and, hence, $D_{\min}\subseteq D(T_{[{\cal A}]})$. 
Let ${\cal A}(c_j)$ be a $d\times 2$ matrix of the form $\begin{pmatrix}a_{ij}(c_j)&b_{ij}(c_j)\end{pmatrix}$. In each interval $I_j$ we utilize \cite[Lem. 5.5.7]{GNZ23} to infer that there exist functions $h_{i,j}\in D_{\max}^{(j)}$ with $i\in\{1,\ldots, d\}$ such that 
\begin{align}\label{46}
  \bar{a}_{ij}(c_j)=\sigma(c_j)W_{j}(\hatt{u}_{c_j},h_{i,j})(c_j),\quad \bar{b}_{ij}(c_j)=\sigma(c_j)W_{j}(u_{c_j},h_{i,j})(c_j).   
 \end{align}
Let $T$ be the extension defined by the condition \eqref{21}. By using the functions $\{\bf h_i\}\in D_{\max}$, leading to \eqref{46}, in the condition \eqref{21}, we find that the boundary conditions defining $T$ coincides with those determined by the matrices ${\cal A}(c_j)$. Therefore we can conclude that $D(T)=D(T_{[{\cal A}]})$. It remains to show that the functions $\{\bf h_i\}\in D_{\max}$, chosen to obtain \eqref{46}, form a generalized boundary set. The functions $\{\bf h_i\}\in D_{\max}$ are linearly independent modulo $D_{\min}$. In fact, let ${\bf g}=c_{1}{\bf h}_{1}+\cdots+c_{d}{\bf h}_{d}\in D_{\min}$ , with $c_i\in\C$, then for any ${\bf f}\in D_{\max}$ we have $\sigma(c_j)W_{j}(f_j,g_j)(c_j)=0$ for $j\in{\cal S}$. By choosing $\hatt{\bf U}=(\hatt U_1,\ldots,\hatt U_{N+1})$ and ${\bf U}=(U_1,\ldots,U_{N+1})$ in $D_{\max}$ as in \eqref{35}, these conditions give, explicitly, 
\begin{align}
    \sum_{i=1}^{d}\bar{c}_{i}a_{ij}(c_j)=0,\quad \sum_{i=1}^{d}\bar{c}_{i}b_{ij}(c_j)=0,\quad j\in{\cal S},\; c_{j}=\{a_j,b_j\}.
\end{align}
 These, in turn, can be written in matrix form as 
 \begin{align}
     \begin{pmatrix}
        {\cal A}^{T}(a_1)\\
        {\cal A}^{T}(b_1)\\
        \vdots\\
        {\cal A}^{T}(a_{N+1})\\
        {\cal A}^{T}(b_{N+1})\\
    \end{pmatrix}\begin{pmatrix}
        \bar{c}_1\\
        \bar{c}_2\\
        \vdots\\
        \bar{c}_d
    \end{pmatrix}=0.
 \end{align}
Because of the rank condition \eqref{23a}, the above linear system has only the trivial solution, and hence, the set $\{{\bf h}_{i}\}$ is linearly independent modulo $D_{\min}$. Finally, condition \eqref{23b} implies that 
\begin{align}
 \sum_{j=1}^{N+1}\sum_{c_j\in\{a_j,b_j\}}\sigma(c_j)\left(\bar{a}_{ij}(c_j)b_{kj}(c_j)-\bar{b}_{ij}(c_j)a_{kj}(c_j)\right)=0,\quad i,k\in\{1,\ldots,d\}.   
\end{align}
Since \eqref{46} holds, the above equation implies via \eqref{Eq:Plucker} that 
\begin{align}
     \sum_{j=1}^{N+1}\sum_{c_j\in\{a_j,b_j\}}\sigma(c_j)W_{j}(h_{i,j},\bar{h}_{k,j})(c_j)=0,\quad i,k\in\{1,\ldots,d\},
\end{align}
which, in turn, guarantees that condition $(iii)$ of Definition \ref{def2.2} is satisfied.  
 \end{proof}

\begin{remark}
    Quantum graphs are a particular case of an $N+1$ intervals system. In a metric graph, intervals (bonds) are connected at the endpoints (vertices). A quantum graph is the pair consisting of a metric graph and a Laplace operator acting on functions defined on the bonds of the graph. The self-adjoint extensions of quantum graphs are particular cases of those provided in Theorem \ref{Thm2.8}.   
\end{remark}

Theorem \ref{Thm2.8} describes the self-adjoint extensions of a very general system in which the intervals can have any relative position (even overlapping) and each interval can have a different Sturm--Liouville differential expression. Since the main goal of this work is to study generalized point potentials, we restrict our analysis to the case in which the intervals are \emph{adjacent} to each other.
\begin{assumption}
For $n\in\{1,\ldots,N\}$ the intervals $I_n$ and $I_{n+1}$ are adjacent to each other, that is $b_{n}\equiv a_{n+1}$. 
\end{assumption}

\begin{definition}(Interior Boundary)\label{intbound}
    Let $n \in \{1,\ldots,N\}$. The common boundary point $R_n = b_n \equiv a_{n+1}$ 
    of two adjacent intervals $I_n$ and $I_{n+1}$ is called an interior boundary 
 if $R_n$ is a limit circle endpoint for both $\tau_n$ on $I_n$ and 
    $\tau_{n+1}$ on $I_{n+1}$, and the self-adjoint extensions in Theorem \ref{Thm2.8} satisfy the following condition: There exist two distinct rows $r_i$ and $r_j$, $i,j\in \{1,\ldots d\}$ such that \\[1mm] 
$(i)$ $r_i$ and $r_j$ are the only rows of  ${\cal A}(b_n)$ and ${\cal A}(a_{n+1})$ that can be non-vanishing. \\[1mm] 
$(ii)$ All remaining matrices ${\cal A}(c_j)$, $c_j\notin\{b_n,a_{n+1}\}$, have zero rows $r_i$ and $r_j$.   
\end{definition}

\begin{remark}
    This is stating that the boundary conditions to the left and right of $R_n$ involve only the generalized boundary values adjacent to $R_n$ and no endpoint values from the remaining intervals. 
\end{remark}

We can now introduce the notion of generalized point potential as a particular case of interior boundary point. 
\begin{definition}(Generalized Point Potential)\label{def}
A generalized point potential at $R_n = b_n \equiv a_{n+1}$ is an interior boundary as in Definition \ref{intbound} for which the $2\times 2$ sub-matrices of  ${\cal A}(b_n)$ and ${\cal A}(a_{n+1})$ formed by the rows $r_{i}$ and $r_j$ have non-zero determinant. 
\end{definition}

The boundary conditions defining a generalized point potential take a more familiar form once the condition \eqref{23b} is taken explicitly into account. Since a self-adjoint extension is invariant under rearrangements of the $d$ conditions in \eqref{23}, we can assume without loss of generality that the non-vanishing rows of ${\cal A}(b_n)$ and ${\cal A}(a_{n+1})$ are consecutive, that is, $r_i$ and $r_{i+1}$. In this case the matrices ${\cal A}(b_n)$ and ${\cal A}(a_{n+1})$ are block matrices with a non-degenerate $2\times2$ matrix formed by the rows $r_i$ and $r_{i+1}$ and zeros everywhere else. We call $P(b_n)$ the non-degenerate $2\times 2$ sub-matrix of ${\cal A}(b_n)$ and $P(a_{n+1})$ that of ${\cal A}(a_{n+1})$. The matrices ${\cal A}(c_j)$ with $c_j\notin\{b_n,a_{n+1}\}$ have a complementary block matrix structure with vanishing rows $r_i$ and $r_{i+1}$ and, in general, non-vanishing remaining rows. The matrix ${\cal A}(c_j)E{\cal A}^{\ast}(c_j)$ can, hence, be computed by using block matrix product. The product ${\cal A}(b_n)E{\cal A}^{\ast}(b_n)$ is a  
$d\times d$ matrix that is zero everywhere except for a non-degenerate $2\times 2$ sub-matrix $ P(b_n) E P^{\ast}(b_n)$ positioned in the $i$-th and $(i+1)$-th columns and the $i$-th and $(i+1)$-th rows. The same occurs for the product ${\cal A}(a_{n+1})E{\cal A}^{\ast}(a_{n+1})$ except that in this case the non-degenerate $2\times 2$ sub-matrix is $P(a_{n+1}) E P^{\ast}(a_{n+1})$. All other products ${\cal A}(c_j)E{\cal A}^{\ast}(c_j)$ with $c_j\notin\{b_n,a_{n+1}\}$ give, instead, $d\times d$ matrices that have a vanishing $2\times 2$ block formed by the $i$-th and $(i+1)$-th columns and the $i$-th and $(i+1)$-th rows. Since ${\cal A}(b_n)$ and ${\cal A}(a_{n+1})$ have the block structure outlined above, the conditions \eqref{23} yield, at the interior boundary $R_n=b_n\equiv a_{n+1}$, the following relation 
\begin{align}\label{51}
    P(b_n)\begin{pmatrix}
	     \wti{f}_{n}(b_n) \\
         \wti{f}'_{n}(b_n)
	 \end{pmatrix}=P(a_{n+1})\begin{pmatrix}
	     \wti{f}_{n+1}(a_{n+1}) \\
         \wti{f}'_{n+1}(a_{n+1})
	 \end{pmatrix},
\end{align}
while condition \eqref{23b} at $R_n=b_n\equiv a_{n+1}$ leads to
\begin{align}\label{52}
  P(b_n) E P^{\ast}(b_n)  = P(a_{n+1}) E P^{\ast}(a_{n+1}).
\end{align}
The constraint in \eqref{52} implies that the boundary condition in \eqref{51} can be rewritten as \cite[Thm. 5.5.8]{GNZ23}
\begin{align}
  \begin{pmatrix}
	     \wti{f}_{n+1}(a_{n+1}) \\
         \wti{f}'_{n+1}(a_{n+1})
	 \end{pmatrix}=e^{i\varphi_n}M_n\begin{pmatrix}
	     \wti{f}_{n}(b_n) \\
         \wti{f}'_{n}(b_n)
	 \end{pmatrix},\quad \varphi_n\in[0,\pi),\quad M_n\in SL(2,\R).  
\end{align}
This argument serves as a proof of the following characterization of a generalized point potential:
\begin{theorem}
    A generalized point potential at $R_n$ is an interior boundary characterized by the matching conditions
    \begin{align}
  \begin{pmatrix}
	     \wti{f}_{n+1}(R_n^{(+)}) \\
         \wti{f}'_{n+1}(R_n^{(+)})
	 \end{pmatrix}=e^{i\varphi_n}M_n\begin{pmatrix}
	     \wti{f}_{n}(R_n^{(-)}) \\
         \wti{f}'_{n}(R_n^{(-)})
	 \end{pmatrix},\quad \varphi_n\in[0,\pi),\quad M_n\in SL(2,\R).  
\end{align}
\end{theorem}

\begin{remark}
    The generalized point potential defined above contains the well-known $\delta$, $\delta'$, and $\delta$-$\delta'$ potentials as particular cases. 
\end{remark}

The self-adjoint extensions of the minimal operator $T_{\min}$ in the setting of $N+1$ adjacent intervals with $N$ generalized point potentials can be organized in three types, each characterized by the nature of the two extremal endpoints of the chain of intervals. In particular, the self-adjoint extensions of first type consist of those for which both endpoints $a_1$ and $b_{N+1}$ are limit points. In the second type, one of the two endpoints, either $a_1$ or $b_{N+1}$ is a limit point and the other is a limit circle point. Lastly, the self-adjoint extensions of third type are those for which both endpoints are limit circle. 

For the sake of simplicity, we will denote the left endpoint of the chain of adjacent intervals by $a$ and the right endpoint by $b$. Letting ${\cal Q}=\{1,\ldots, N\}$, then the generalized point potentials are positioned at $R_{i}$, $i\in\cal Q$. We can now describe in detail the three types of self-adjoint extensions of $T_{\min}$ in the setting of $N+1$ adjacent intervals with $N$ generalized point potentials. 

\begin{paragraph}{Two limit point endpoints.}
This class is characterized by both endpoints $a$ and $b$ being limit point. More precisely, when both endpoints $a$ and $b$ are limit point, the dynamics of the scalar field under the influence of $N$ generalized point potentials is described by the self-adjoint extensions
\begin{align}\label{55e}
         & T_{[M]}\,{\bf f}=\tau{\bf f},\nonumber\\
    &{\bf f}\in D\left(T_{[M]}\right)=\left\{{\bf f}\in D_{\max}\Bigg|\, \begin{pmatrix}
	     \wti{f}_{j+1}(R_j^{(+)}) \\
         \wti{f}'_{j+1}(R_j^{(+)})
	 \end{pmatrix}=e^{i\varphi_j}M_j\begin{pmatrix}
	     \wti{f}_{j}(R_j^{(-)}) \\
         \wti{f}'_{j}(R_j^{(-)})
	 \end{pmatrix},\;j\in{\cal Q}\right\},
    \end{align}
with $[{M}]=(e^{i\varphi_1}M_1,\ldots,e^{i\varphi_N}M_N)$, $M_j\in SL(2,\R)$, and $\varphi_j\in[0,\pi)$. 
\end{paragraph}
\begin{paragraph}{One limit point endpoint.}
  We can assume, without loss of generality, that the right endpoint $b$ is limit point. In this case,
the extensions that describe $N$ generalized point potentials with a limit circle endpoint at $a$ and a limit point at $b$ are
\begin{align}\label{56}
         & T_{\alpha,[M]}\,{\bf f}=\tau{\bf f},\quad {\bf f}\in D\left(T_{\alpha,[M]}\right),\nonumber\\
    &D\left(T_{\alpha,[M]}\right)=\left\{{\bf f}\in D_{\max}\Bigg|\,\cos\alpha  \wti{f}_{1}(a)+\sin\alpha\,  \wti{f}'_{1}(a)=0,\; \begin{pmatrix}
	     \wti{f}_{j+1}(R_j^{(+)}) \\
         \wti{f}'_{j+1}(R_j^{(+)})
	 \end{pmatrix}=e^{i\varphi_j}M_j\begin{pmatrix}
	     \wti{f}_{j}(R_j^{(-)}) \\
         \wti{f}'_{j}(R_j^{(-)})
	 \end{pmatrix},\;j\in{\cal Q}\right\},
    \end{align}
with, once again, $M_j\in SL(2,\mathbb{R})$, $\varphi_j\in[0,\pi)$, and $\alpha\in[0,\pi)$.    
\end{paragraph}
\begin{paragraph}{No limit point endpoints.}
If neither endpoint is limit point, the $N$ generalized point potential system can be further divided into two mutually exclusive subclasses. In one case, the two limit circle endpoints are \emph{separated} while in the other case they are \emph{coupled}. The separated subclass is described by the following self-adjoint extensions 
\begin{align}\label{57}
         & T_{\alpha,\beta,[{M}]}\,{\bf f}=\tau{\bf f},\quad {\bf f}\in D\left(T_{\alpha,\beta,[M]}\right),\nonumber\\
    &D\left(T_{\alpha,\beta,[M]}\right)\nonumber\\
    &\quad=\left\{{\bf f}\in D_{\max}\Bigg|\,\begin{array}{lr}
\cos\alpha \wti{f}_{1}(a)+\sin\alpha\, \wti{f}'_{1}(a)=0&\\
\cos\beta \wti{f}_{N+1}(b)-\sin\beta\, \wti{f}'_{N+1}(b)=0
\end{array}\!\!\!\!\!\!\!\!, \begin{pmatrix}
	     \wti{f}_{j+1}(R_j^{(+)}) \\
         \wti{f}'_{j+1}(R_j^{(+)})
	 \end{pmatrix}=e^{i\varphi_j}M_j\begin{pmatrix}
	     \wti{f}_{j}(R_j^{(-)}) \\
         \wti{f}'_{j}(R_j^{(-)})
	 \end{pmatrix},\;j\in{\cal Q}\right\},
    \end{align}
    with $\alpha,\beta\in[0,\pi)$. The coupled subclass is, instead, characterized by 
\begin{align}\label{58}
         & T_{P,\eta,[{M}]}\,{\bf f}=\tau{\bf f},\quad {\bf f}\in D\left(T_{P,\eta,[M]}\right),\nonumber\\
    & D\left(T_{P,\eta,[M]}\right)=\left\{{\bf f}\in D_{\max}\Bigg|\begin{pmatrix}
	     \wti{f}_{N+1}(b) \\
         \wti{f}'_{N+1}(b)
	 \end{pmatrix}=e^{i\eta}P\begin{pmatrix}
	     \wti{f}_{1}(a) \\
         \wti{f}'_{1}(a)
	 \end{pmatrix}, \begin{pmatrix}
	     \wti{f}_{j+1}(R_j^{(+)}) \\
         \wti{f}'_{j+1}(R_j^{(+)})
	 \end{pmatrix}=e^{i\varphi_j}M_j\begin{pmatrix}
	     \wti{f}_{j}(R_j^{(-)}) \\
         \wti{f}'_{j}(R_j^{(-)})
	 \end{pmatrix},\;j\in{\cal Q}\right\},
    \end{align} 
    with $P\in SL(2,\R)$, and $\eta\in[0,\pi)$. The class containing no limit points and only one generalized point potential arises naturally when studying, for instance, scalar fields in piston configurations which are widely used models
for understanding the properties of the Casimir force (see e.g. \cite{fucci15,fucci21} and references therein).
\end{paragraph}

\section{Spectral zeta function}\label{Sec3}

In the paper \cite{FPS25a} we have shown how to construct and analyze the $\z$-function associated with a sequence of complex numbers $\{\mu_i\}_{i\in\N}$ satisfying the conditions 
\cite[Def. 2.1]{FPS25a}:\\[1mm]
$(i)$ $0<|\mu_1|\leq|\mu_2|\leq \dots \to\infty$;\\[1mm]
    $(ii)$ the exponent of convergence $\kappa=\inf\big\{\rho > 0 \, \colon \sum_{n\in\N} |\mu_{n}|^{-\rho} < \infty \big\}$ is finite;\\[1mm]
    $(iii)$  there exists $\varepsilon>0$ and $\Psi\in[-\pi,\pi)$ such that all $\mu_n\in \mathcal{I}_\varepsilon=\left\{z\in\C : \textrm{Arg}(z) \not \in (\Psi-\varepsilon, \Psi+\varepsilon)\right\}$.\\[1mm]
In our $N+1$ interval framework, the endpoints are allowed to be limit point. This implies that the Sturm--Liouville operator on the interval with the limit point endpoint could have a continuous part to its spectrum. Moreover, even if a continuous part is absent, the spectrum could grow too slowly leading to an infinite exponent of convergence. To avoid these situations, for which a $\z$-function is not well-defined, we make the following assumptions:
\begin{assumption}\label{h3a}
$(I)$ The self-adjoint extensions $T_{[M]}$, $T_{\alpha,[M]}$, $T_{\alpha,\beta,[M]}$, and $T_{P,\eta,[M]}$ of $T_{\min}$, denoted by $T_i$, have only discrete spectrum $\sigma(T_i)=\{\lambda_{i,n}\}_{n\in \N}$ where the eigenvalues are counted according to their multiplicity.\\[1mm]
$(II)$ The exponents of convergence $\kappa_i=\limsup_{n\to\infty}(\ln \, n)/(\ln|\lambda_{i,n}|)$ associated with $\sigma(T_i)$ are finite.
\end{assumption}  
The first assumption, together with the fact that $T_{\min}$ is bounded from below, see Remark \ref{rem2.6}, ensures that conditions $(i)$ and $(iii)$ are satisfied, ignoring any zero eigenvalues. The second assumption is simply an equivalent way of stating condition $(ii)$ for the nonzero spectrum. 

When Assumption \ref{h3a} holds, one can define the spectral $\zeta$-function associated with the spectrum $\sigma(T_i)=\{\lambda_{i,n}\}_{n\in \N}$ as 
\begin{equation}\label{zeta}
    \zeta(s;T_i)=\sum_{\underset{\lambda_{i,n}\neq0}{n\in\N}}\lambda_{i,n}^{-s}, \quad \Re(s)>\kappa,
\end{equation}
where, unlike in \cite{FPS25a}, we cannot assume that there are no zero eigenvalues, so are removed from the sum.

To analyze the structure and properties of the spectral $\z$-function we rely on a suitable contour integral representation. It was shown in \cite[Thm. 2.4]{FPS25a} that if the spectrum satisfies the Assumption \ref{h3a} (and hence the conditions $(i)$--$(iii)$ above) then the $\z$-function in \eqref{zeta} has the integral representation 
\begin{align}\label{intzetaH}
    \zeta(s;T_i)=\frac{1}{2\pi i}\int_\gamma dz\,z^{-s}\frac{d}{dz}\ln\big[ z^{-m_{0,i}} H_i(z)\big],\quad \Re(s)>\kappa,
\end{align}
where $m_{0,i}$ represents the multiplicity of the zero eigenvalue, $\gamma$ is a contour in the complex plane enclosing the spectrum $\sigma(T_i)$ in the counterclockwise direction dipping below, hence avoiding, the origin with appropriately chosen branch cut $($see \cite[Fig. 1]{FPS25}$)$, and $H_i(z)$ is the \emph{Hadamard characteristic function} defined, with $k\in\N$ and $k-1\leq\kappa<k$, as
\begin{equation}\label{c1}
   H_i(z)=z^{m_{0,i}}\prod_{n=1}^{\infty}E\left(\frac{z}{\lambda_{i,n}},k-1\right),\quad \text{ with }\quad 
  E\left(\frac{z}{\lambda_{i,n}},k-1\right)=\left(1-\frac{z}{\lambda_{i,n}}\right)\exp\left[\sum_{j=1}^{k-1}\frac{1}{j}\left(\frac{z}{\lambda_{i,n}}\right)^{j}\right].
\end{equation}
While the expression in \eqref{intzetaH} provides an integral representation for the spectral $\zeta$-function, it is restrictive since it relies on the explicit knowledge of the Hadamard characteristic function. This problem can be circumvented by utilizing a more general notion of characteristic function as suggested in \cite{FPS26}. In fact, any entire function $F_{i}(z)$ with finite order $\rho_i$ and zeros exactly at $\sigma(T_i)=\{\lambda_{i,n}\}_{n\in \N}$, counting multiplicities, can be used as a characteristic function in the integral representation \eqref{intzetaH} of the spectral $\z$-function associated with the spectrum $\sigma(T_i)$. More precisely, one can prove the following:
\begin{theorem}[{\cite[Thm. 1.3]{FPS26}}]\label{Thrm3.4}
    Let Assumption \ref{h3a} be satisfied. If a characteristic function $F_{i}(z)$ associated with the spectrum $\sigma(T_i)$ exists and has finite order $\rho_i$, then for $\Re(s)>\rho_i\geq\kappa_i$,
\begin{align}\label{3.2}
\zeta(s;T_i)&=\frac{1}{2\pi i}\int_\gamma dz \, z^{-s} \frac{d}{dz} \ln\left[z^{-m_{0,i}}F_{i}(z)\right]\\
&=e^{is(\pi-\Psi)}\frac{\sin(\pi s)}{\pi}\int_{\delta}^{\infty}dt\,t^{-s}\frac{d}{dt}\ln\left[\left(te^{i\Psi}\right)^{-m_{0,i}}F_{i}\left(te^{i\Psi}\right)\right]-\frac{1}{2\pi i}\int_{C_\delta}dz\,z^{-s}\frac{d}{dz}\ln\left[z^{-m_{0,i}}F_{i}(z)\right],\notag
\end{align}
where $m_{0,i}$ represents the multiplicity of the zero eigenvalue, $C_\delta$ is a clockwise circle with radius $0<\delta<\min\{|\lambda_{i,n}|: \lambda_{i,n}\neq0\}$  parametrized via $z=\delta e^{i\theta}$ from $\theta=\Psi$ to $\theta=\Psi-2\pi$, $\Psi\in (\pi/2,\pi)$ with $R_\Psi=\{z=te^{i\Psi} \, | \, t\in [0,\infty)\},$ $\Psi\in (\pi/2,\pi)$ denoting the branch cut, and $\gamma$ is a counterclockwise contour which encloses the spectrum $\sigma(T_{i})$, dipping below, hence avoiding, the origin $($see \cite[Fig. 1]{FPS25}$)$.  
\end{theorem}

To analyze the spectral $\z$-function further, we need to obtain a characteristic function for each of the self-adjoint extensions defining the generalized point potentials introduced in Section \ref{Sec2}.

\subsection{Characteristic function associated with the generalized point potentials}\label{Sec3.1}

In this section we construct the characteristic function for the three types of self-adjoint extensions describing the generalized point potentials under Assumption \ref{h3a}. We consider, first, the case of two limit point endpoints. 

In the intervals $I_{j}$, $j\in\{2,\ldots,N\}$, we introduce an entire fundamental system of solutions $\theta_{j}(z,x,R_{j-1})$ and $\phi_{j}(z,x,R_{j-1})$ of $\tau_j y(z,x)=z y(z,x)$ defined by\footnote{Such a normalized entire fundamental system exists due to each $R_j$ being limit circle nonoscillatory.}
\begin{align}\label{63}
\wti\theta_j(z,R_{j-1},R_{j-1})=\wti\phi_j^{\, \prime}(z,R_{j-1},R_{j-1})=1,\quad \wti\theta_j^{\, \prime}(z,R_{j-1},R_{j-1})=\wti\phi_j(z,R_{j-1},R_{j-1})=0,
\end{align}
such that
\begin{align}\label{64}
    W_j\big(\theta_j(z,\dott,R_{j-1}),\phi_j(z,\dott,R_{j-1})\big)=1.
\end{align}
In the first interval, $I_{1}$, we consider an entire principal solution $u_a(z,x)$ near $a$ with finite growth order in $z$ and, similarly, in the last interval, $I_{N+1}$, we consider an entire principal solution $u_b(z,x)$ near $b$ with finite growth order in $z$.\footnote{Entire solutions exist from discrete spectrum assumptions by \cite[Lem.~2.2, 2.4]{KST_IMRN}, while existence of finite growth order (even minimal growth order) principal solutions is a consequence of Assumption \ref{h3a} $(II)$ by \cite[Sec. 3.3]{FPS26}.}

For $j\in\{2,\ldots,N\}$ and $c_j$ either of the endpoints of the interval $I_j$, we introduce the generalized boundary value matrix as
\begin{align}\label{65}
    \Phi_{j}(c_j)=\begin{pmatrix}
        \wti\theta_j(z,c_j,R_{j-1})  & \wti\phi_j(z,c_j,R_{j-1})\\
        \wti\theta_j^{\, \prime}(z,c_j,R_{j-1}) & \wti\phi_j^{\, \prime}(z,c_j,R_{j-1})
    \end{pmatrix}.
\end{align}
Note that $\Phi_{j}(R_{j-1})=\mathbb{I}$. 
The conditions in \eqref{55e} can be written in terms of generalized boundary value matrices as follows by setting $f_{j+1}=\beta_{j+1}\phi_{j+1}+\alpha_{j+1}\theta_{j+1}$ for $j=2,\dots, N-1$: 
\begin{align}\label{66}
    \begin{pmatrix}
        \alpha_2\\
        \beta_2
    \end{pmatrix}&=e^{i\varphi_1}M_1\begin{pmatrix}
        \alpha_1 \wti u_{a}(z,R_1)\\
        \alpha_1 \wti u'_{a}(z,R_1)
    \end{pmatrix},\nonumber\\
    \begin{pmatrix}
        \alpha_{j+1}\\
        \beta_{j+1}
    \end{pmatrix}&=e^{i\varphi_j}M_j\Phi_{j}(R_j)\begin{pmatrix}
        \alpha_j\\
        \beta_j
    \end{pmatrix},\quad j\in\{2,\ldots,N-1\},\nonumber\\
   \begin{pmatrix}
        \alpha_{N+1} \wti u_{b}(z,R_N)\\
        \alpha_{N+1} \wti u'_{b}(z,R_N)
    \end{pmatrix}&=e^{i\varphi_N}M_N\Phi_{N}(R_N)\begin{pmatrix}
        \alpha_N\\
        \beta_N
    \end{pmatrix},
\end{align}
where $\alpha_k\in\C$ and $\beta_k\in\C$ are unknowns. The linear system \eqref{66} of $2N$ equations has the form of a recursive relation and can be reduced to one in two equations and two unknowns as follows
\begin{align}\label{69}
\begin{pmatrix}
        \alpha_{N+1} \wti u_{b}(z,R_N)\\
        \alpha_{N+1} \wti u'_{b}(z,R_N)
    \end{pmatrix}&=e^{i\Theta_N}{\cal T}_N\begin{pmatrix}
        \alpha_1 \wti u_{a}(z,R_1)\\
        \alpha_1 \wti u'_{a}(z,R_1)
    \end{pmatrix},   
\end{align}
where $\Theta_N=\varphi_N+\cdots+\varphi_1$ and ${\cal T}_{N}$ is defined as
\begin{align}\label{70}
  {\cal T}_N=M_N\Phi_{_{N}}(R_{N})M_{N-1}\Phi_{_{N-1}}(R_{N-1})\cdots M_2\Phi_{_{2}}(R_{2})M_1,  
\end{align}
and represents a \emph{transfer matrix} which propagates generalized boundary data from the first interior boundary to the last and, hence, encodes the cumulative effect of the $N$ point potentials.  
The linear system in \eqref{69} has a non-trivial solution if the entire (by construction) function 
\begin{align}\label{71}
 F_{[M]}(z) =\wti u_{b}(z,R_N)[{\cal T}_{N,21}\wti u_{a}(z,R_1)+{\cal T}_{N,22}\wti u'_{a}(z,R_1)]-\wti u'_{b}(z,R_N) [{\cal T}_{N,11}\wti u_{a}(z,R_1)+{\cal T}_{N,12}\wti u'_{a}(z,R_1)],
\end{align}
vanishes.\footnote{Taking $N=1$, $\varphi_1=0$, $M_1=\mathbb{I}$, and $\tau_1=\tau_2$, this characteristic function becomes $F(z)=W(u_b(z,\dott),u_a(z,\dott))(R_1)=\lim_{x\uparrow b}u_a(z,x)/\hatt{u}_b(z,x)$ by linear independence, recovering \cite[Eq. (3.54)]{FPS26} up to choosing minimal growth principal solutions.} As each of the interior boundary points are limit circle, the endpoint-normalized fundamental solutions $\theta_{j}(z,x,R_{j-1})$ and $\phi_{j}(z,x,R_{j-1})$ are entire functions of $1/2$-order in the spectral parameter $z$. Moreover, one can construct the principal solutions $u_a(x,z)$ and $u_b(x,z)$ so that their order is finite, or even equal to the exponent of convergence of the spectrum \cite[Sec. 3.3]{FPS26}. This implies, together with the fact that the zeros of  $F_{[M]}(z)$ coincide with the spectrum $\sigma(T_{[M]})$, that the function $F_{[M]}(z)$ above is a characteristic function for the system consisting of $N$ generalized point potentials with limit point endpoints.  

The case of a single limit point endpoint is handled in the same manner as above, with the sole modification that the Weyl solution is considered at the limit circle endpoint. More precisely, let us assume, without loss of generality, that the endpoint $a$ is limit circle. Then one considers, in the interval $I_1$, the entire Weyl solution of $\tau_{1}y(z,x)=zy(z,x)$ at $a$ defined by
\begin{align}\label{72}
    \wti \Psi_{\a}(z,a)=-\sin\alpha,\quad \wti \Psi'_{\a}(z,a)=\cos\alpha.
\end{align}
By imposing the conditions \eqref{56} one obtains a linear system of the form \eqref{66} with $\wti u_{a}(z,R_1)$ and $\wti u'_{a}(z,R_1)$ replaced with $\wti \Psi_{a}(z,R_1)$ and $\wti \Psi'_{a}(z,R_1)$, respectively. An argument analogous to the one outlined above leads to the characteristic function
\begin{align}
F_{\alpha,[M]}(z) =\wti u_{b}(z,R_N)[{\cal T}_{N,21}\wti \Psi_{a}(z,R_1)+{\cal T}_{N,22}\wti \Psi'_{a}(z,R_1)]-\wti u'_{b}(z,R_N) [{\cal T}_{N,11}\wti \Psi_{a}(z,R_1)+{\cal T}_{N,12}\wti \Psi'_{a}(z,R_1)],    
\end{align}
for the system of $N$ generalized point potentials with one limit point endpoint. 

The same argument can be used for the case in which the endpoints are limit circles and are separated. In this setting, we consider the entire Weyl solutions at both limit circle endpoints, that is the solution already defined in \eqref{72} and the solution of $\tau_{N+1}y(z,x)=zy(z,x)$ at $b$ defined by
\begin{align}\label{72a}
    \wti \Psi_{\b}(z,b)=\sin\b,\quad \wti \Psi'_{\b}(z,b)=\cos\b.
\end{align} 
By imposing the conditions in \eqref{57} one obtains a linear system analogous to the one in \eqref{66} but with $\wti u_{a}(z,R_1)$ and $\wti u'_{a}(z,R_1)$ replaced with $\wti \Psi_{a}(z,R_1)$ and $\wti \Psi'_{a}(z,R_1)$, respectively and $\wti u_{b}(z,R_N)$ and $\wti u'_{b}(z,R_N)$ replaced with $\wti \Psi_{\b}(z,R_N)$ and $\wti \Psi'_{\b}(z,R_N)$, respectively. This leads, as before, to the characteristic function
\begin{align}\label{73}
F_{\alpha,\b,[M]}(z) =\wti \Psi_{\b}(z,R_N)[{\cal T}_{N,21}\wti \Psi_{a}(z,R_1)+{\cal T}_{N,22}\wti \Psi'_{a}(z,R_1)]-\wti \Psi'_{\b}(z,R_N) [{\cal T}_{N,11}\wti \Psi_{a}(z,R_1)+{\cal T}_{N,12}\wti \Psi'_{a}(z,R_1)],    
\end{align}
for the system of $N$ generalized point potentials with two separated limit circle endpoints. 

Lastly, we consider the case of two coupled limit circle endpoints. In all intervals $I_j$, $j\in\{1,\ldots,N+1\}$, we introduce a fundamental system of solutions of $\tau_j y(z,x)=z y(z,x)$ satisfying the normalization conditions \eqref{63} and \eqref{64} at the left endpoints of each interval. In addition, the definition \eqref{65} of the boundary value matrix will be extended to the first and last intervals. Then the conditions in \eqref{58} lead to the linear system
\begin{align}\label{76}
   \Phi_{N+1}(b) \begin{pmatrix}
        \alpha_{N+1}\\
        \beta_{N+1}
    \end{pmatrix}&=e^{i\eta}P\begin{pmatrix}
        \alpha_1\\
        \beta_1
    \end{pmatrix},\nonumber\\
    \begin{pmatrix}
        \alpha_{l+1}\\
        \beta_{l+1}
    \end{pmatrix}&=e^{i\varphi_l}M_l\Phi_{l}(R_l)\begin{pmatrix}
        \alpha_l\\
        \beta_l
    \end{pmatrix},\quad l\in\{1,\ldots,N\}.
\end{align}
By solving using recursion, the system in \eqref{76} reduces to the following
 \begin{align}\label{74}
     \left[e^{i\Theta_N}\Phi_{N+1}(b){\cal T}_N\Phi_{1}(R_1)-e^{i\eta}P\right]\begin{pmatrix}
        \alpha_{1}\\
        \beta_{1}
    \end{pmatrix}=0  
    \end{align}
The system has a non-trivial solution if the following characteristic function vanishes:
 \begin{align}\label{78}
      F_{P,[M]}(z)=\det\left[e^{i\Theta_N}\Phi_{N+1}(b){\cal T}_N\Phi_{1}(R_1)-e^{i\eta}P\right]. 
    \end{align}
\begin{remark}
    When $N=1$, $\varphi_1=0$, $M_1=\mathbb{I}$, and $\tau_1=\tau_2$ the characteristic function \eqref{78} agrees with the ones given in \cite[Eq. (2.17)]{FGKS21} for $\tau$ regular and \cite[Eq. (3.17)]{FPS25} for $\tau$ quasi-regular. In fact, in this case we get
    \begin{align}\label{75a}
      F_{P,\mathbb{I}}(z)=\det\left[\Phi_{2}(b)\Phi_{1}(R)-e^{i\eta}P\right].
    \end{align}
The fundamental set of solutions in one interval can be written in terms of the fundamental set of solutions in the other interval as
\begin{align}
    \Phi_{2}(x)=\Phi_{1}(x)\mathcal{C},\quad \textrm{with}\quad 
    \mathcal{C}=\begin{pmatrix}
        \phi_{1}^{[1]}(z,R,a)& -\phi_{1}(z,R,a)\\
        -\theta_{1}^{[1]}(z,R,a)& \theta_{1}(z,R,a)
    \end{pmatrix}.
\end{align}
Since $\Phi_{2}(R)=\mathbb{I}$, we have that $\Phi_{1}(R)=\mathcal{C}^{-1}$ and therefore
\begin{align}
    \Phi_{2}(b)\Phi_{1}(R)=\Phi_{1}(b)\mathcal{C}\Phi_{1}(R)=\Phi_{1}(b).
\end{align}
Therefore, \eqref{75a} reduces to 
\begin{align}
     F_{P,\mathbb{I}}(z)=\det\left[\Phi_{1}(b)-e^{i\eta}P\right],
\end{align}
which agrees with \cite[Eq. (2.17)]{FGKS21} for $\tau$ regular and \cite[Eq. (3.17)]{FPS25} for $\tau$ quasi-regular.
\end{remark}

The spectral zeta function for each of these $N$ generalized point potentials settings is obtained by substituting the corresponding characteristic function in the expression \eqref{3.2}. According to Theorem \ref{Thrm3.4}, the integral representation of the spectral zeta function associated with the system of generalized point potentials is only valid to the right of the abscissa of convergence $\Re(s)=\rho_i$, the order of $F_{i}(z)$. The goal is to extend this region of convergence by analytic continuation. The analytic continuation procedure has been described previously in the literature, most recently in \cite{FPS25a}, and it relies on the large-$z$ asymptotic expansion of the characteristic function. 
Deriving a general large-$z$ asymptotic expansion for these characteristic functions is impractical for two main reasons: First, the transfer matrix becomes increasingly complicated as the number of point interactions grows. Second, the generalized boundary values entering the characteristic functions depend on the specific differential expressions in each interval.
Before turning to this in detail, we consider the coalescing of generalized point potentials.

\begin{theorem}\label{Thm3.2}(Limiting Behavior)
    Assume hypotheses \ref{Hyp1} and \ref{Hyp2} hold. If $\tau_j$ is quasi-regular in the interval $I_{j}$, $j\in{\cal S}$, then
    \begin{align}\label{79}
      \lim_{R_j\to R_{j-1}}\Phi_{j}(R_j)=\mathbb{I}.  
    \end{align}
\end{theorem}
\begin{proof}
    Let $f_1(z,x)$ and $f_{2}(z,x)$ be solutions of $\tau_{j}y(z,x)=z y(z,x)$ such that $W_{j}(f_1(z,\dott),f_2(z,\dott))=C\neq 0$. By using these functions we can construct a fundamental system of solutions at $R_{j-1}$ as
    \begin{align}
        \phi(z,x,R_{j-1})=\frac{1}{C}\left[\wti f_{1}(z,R_{j-1})f_{2}(z,x)-\wti f_{2}(z,R_{j-1})f_{1}(z,x)\right]\\
        \theta(z,x,R_{j-1})=\frac{1}{C}\left[\wti f'_{2}(z,R_{j-1})f_{1}(z,x)-\wti f'_{1}(z,R_{j-1})f_{2}(z,x)\right].
    \end{align}
Since $\tau_j$ is quasi-regular on $I_{j}$ (both endpoints are limit circle nonoscillatory) we can introduce the regularizing function \cite[Thm. 2.5]{FPS25} 
\begin{align}
    \hatt{u}(\lambda_0,x)=\begin{cases}
        \hatt{u}_{R_{j-1}}(\lambda_0,x) \quad\textrm{for $x$ near $R_{j-1}$},\\
        \hatt{u}_{R_{j}}(\lambda_0,x) \quad\textrm{for $x$ near $R_{j}$},
    \end{cases}
\end{align}
and define
\begin{align}
    v_{i}(z,x)=\frac{f_i(z,x)}{ \hatt{u}(z,x)},\quad i\in\{1,2\},\quad \phi_{r}(z,x,R_{j-1})=\frac{\phi(z,x,R_{j-1})}{ \hatt{u}(z,x)},\quad \theta_r(z,x,R_{j-1})=\frac{\theta(z,x,R_{j-1})}{\hatt{u}(z,x)}.
\end{align}
The functions 
 \begin{align}
        \phi_r(z,x,R_{j-1})=\frac{1}{C}\left[\wti f_{1}(z,R_{j-1})v_{2}(z,x)-\wti f_{2}(z,R_{j-1})v_{1}(z,x)\right]\\
        \theta_r(z,x,R_{j-1})=\frac{1}{C}\left[\wti f'_{2}(z,R_{j-1})v_{1}(z,x)-\wti f'_{1}(z,R_{j-1})v_{2}(z,x)\right],
    \end{align}
are then solutions of the associated regular problem on $I_j$ \cite[Lem. 2.6]{FPS25}
\begin{align}\label{86}
    -(P_j(x)g'(z,x))'+Q_j(x)g(z,x)=zR_j(x)g(z,x),
\end{align}
with 
\begin{align}
P_j(x)=[\hatt{u}(\lambda_0,x)]^{2}p_j(x),\quad R_j(x)= [\hatt{u}(\lambda_0,x)]^{2}r_j(x),\quad Q_j(x)=r_j(x)\hatt{u}(\lambda_0,x)(\tau_j\hatt{u})(\lambda_0,x).   
\end{align}

By setting $d=R_{j}-R_{j-1}$ and by recalling that the solutions $v_1(z,x)$ and $v_{2}(z,x)$ of the regular problem \eqref{86} are in $AC([R_{j-1},R_{j}])$, we can conclude that \cite[Lem. 2.7]{FPS25}
\begin{align}
\lim_{d\to 0^{+}}  \wti\phi(z,R_{j-1}+d,R_{j-1})&=\lim_{d\to 0^{+}} \phi_r(z,R_{j-1}+d,R_{j-1}) \nonumber\\
&=\lim_{d\to 0^{+}}\frac{1}{C}\left[\wti f_{1}(z,R_{j-1})v_{2}(z,R_{j-1}+d)-\wti f_{2}(z,R_{j-1})v_{1}(z,R_{j-1}+d)\right]\nonumber\\
&=\frac{1}{C}\left[\wti f_{1}(z,R_{j-1})v_{2}(z,R_{j-1})-\wti f_{2}(z,R_{j-1})v_{1}(z,R_{j-1})\right]=\phi_r(z,R_{j-1},R_{j-1})\nonumber\\&=\wti\phi(z,R_{j-1},R_{j-1}).
 \end{align}
The same argument can be used to show that 
\begin{align}
\lim_{d\to 0^{+}}  \wti\theta(z,R_{j-1}+d,R_{j-1})=\wti\theta(z,R_{j-1},R_{j-1}). \end{align}

By introducing the notation $g^{[1]}(z,x)=P(x)g'(z,x)$, and by recalling that the solutions $v_1(z,x)$ and $v_{2}(z,x)$ of the regular problem \eqref{86} satisfy and $v^{[1]}_1(z,x),v^{[1]}_{2}(z,x)\in AC([R_{j-1},R_{j}])$ we have
\begin{align}
\lim_{d\to 0^{+}}  \wti\phi'(z,R_{j-1}+d,R_{j-1})&= \lim_{d\to 0^{+}} \phi_{r}^{[1]}(z,R_{j-1}+d,R_{j-1}) \nonumber\\
&=\lim_{d\to 0^{+}}\frac{1}{C}\left[\wti f_{1}(z,R_{j-1})v^{[1]}_{2}(z,R_{j-1}+d)-\wti f_{2}(z,R_{j-1})v^{[1]}_{1}(z,R_{j-1}+d)\right]\nonumber\\
&=\frac{1}{C}\left[\wti f_{1}(z,R_{j-1})v^{[1]}_{2}(z,R_{j-1})-\wti f_{2}(z,R_{j-1})v^{[1]}_{1}(z,R_{j-1})\right]=\phi^{[1]}_r(z,R_{j-1},R_{j-1})\nonumber\\&=\wti\phi'(z,R_{j-1},R_{j-1}).  
\end{align}
The same argument can be utilized to prove that 
\begin{align}
\lim_{d\to 0^{+}}  \wti\theta'(z,R_{j-1}+d,R_{j-1})=\wti\theta'(z,R_{j-1},R_{j-1}). 
\end{align}
The relations \eqref{63} are then used to conclude that \eqref{79} holds.
\end{proof}

By applying this last result recursively, we find that when the $N$ generalized point potentials coalesce we obtain a single point potential described by the matrix 
\begin{align}
    {\cal T}=M_NM_{N-1}\cdots M_1.
\end{align}
It is clear from this expression that when $N$ generalized point potentials merge the resulting point potential is, in general, different than the ones that were previously present. There are, however, cases in which the new point potential is of the same type. In fact, if $M_{i}$, $i\in\{1,\ldots,N\}$, belong to a subgroup of $SL(2,\R)$, then ${\cal T}$ belongs to the same subgroup and, hence, represents the same type of generalized point potential. There are obviously infinitely many subgroups of $SL(2,\R)$, however some of them are of particular interest and are worthy of mention. For instance, when restricting to Schr\"odinger operators, the elements of the maximal unipotent subgroup, which are shear transformations, represent the $\delta'$ potential according to Albeverio and Gesztesy \cite{albeverio}.  The elements of the (upper or lower) Borel subgroup represent, instead, the $\delta$-$\delta'$ potentials (see e.g. \cite{golovaty} and \cite[Sec. 4]{ADK98}). It is worth mentioning that the diagonal subgroup of $SL(2,\R)$ represents the $\delta'$ potential according to Kurasov \cite[Sec. 4]{ADK98}. Theorem  \ref{Thm3.2} implies that coalescing point potentials belonging to the above subgroups give rise to a single point potential in the same subgroup. So, coalescing $\delta$, $\delta'$, and $\delta$-$\delta'$ potentials will give rise to new $\delta$, $\delta'$, and $\delta$-$\delta'$ potentials, respectively, for a fixed $\delta'$ interpretation.      

\subsection{Analytic continuation of the spectral zeta function}

The analytic continuation of the spectral $\z$-function to a region to the left of its abscissa of convergence, has been discussed several times in the literature (see e.g. the monograph \cite{Ki02} and the more recent works \cite{FPS25,FPS25a}) and consists in subtracting, and then adding, a suitable number of terms of the large-$z$ asymptotic expansion of the logarithm of the characteristic function in the integral representation of $\zeta(s;T_i)$.  
Describing this process for general $N$ point potentials is overly complicated and perhaps not particularly illuminating. For this reason, we illustrate the analytic continuation explicitly for two simpler cases that nevertheless retain considerable interest: the first, consists in one generalized point potential with two different expressions, $\tau_1$ and $\tau_2$, defined on the left and the right interval, respectively. This configuration describes interface problems. The second, involves two generalized point potentials with the same operator on each of the three intervals. In this case it is interesting to analyze, for instance, the behavior of the spectral $\z$-function when the point potentials coalesce.       

\subsubsection{Interface problem}

In this case, the differential expression $\tau_1$ is defined on the interval $I_{1}$ while the differential expression $\tau_2$ is defined on $I_2$. The point potential, or the interface, is positioned at $R$ (the common endpoint of $I_1$ and $I_2$). In order for this configuration to be classified as an interface, we must have $\tau_1\neq\tau_2$.

We need to make an assumption regarding the general form of the large-$z$ asymptotic expansion of the generalized boundary 
values that constitute the characteristic functions. 
\begin{assumption}\label{assu1}
    Let $c_i\in\{a,b\}$ denote the endpoints of $I_1\cup I_2$ which can be either limit point or limit circle. Let $f_i(z,x)$ denote the solution which is $L^{2}(I_i,r_idx)$ near $c_i$ when it is a limit point endpoint and the Weyl solution at $c_i$ when it is, instead, a limit circle endpoint. We assume that the generalized boundary value of $f_i(z,x)$ at the interior point $R$ has the large-$z$ asymptotic expansions 
    \begin{align}\label{93}
        \wti f_i(z,R)&=C_i z^{-\chi_i}\exp\left\{\sum_{j=1}^{L_i}d_{j}^{(i)} z^{\beta_{j}^{(i)}}\right\}\sum_{k=0}^{K}A_{k}^{(i)}z^{-\frac{k}{\gamma_i}}+O\left(z^{-\frac{K+1}{\gamma_i}}\right),\\
        \label{94}
        \wti f'_i(z,R)&=\bar{C}_i z^{-\vartheta_i}\exp\left\{\sum_{j=1}^{L_i}d_{j}^{(i)} z^{\beta_{j}^{(i)}}\right\}\sum_{k=0}^{K}B_{k}^{(i)}z^{-\frac{k}{\eta_i}}+O\left(z^{-\frac{K+1}{\eta_i}}\right),
    \end{align}
    with $\chi_i,\beta_{j}^{(i)},\vartheta_i\in\R^{+}$ and $\gamma_i,\eta_i,K,L_i\in\N$.
\end{assumption}
\begin{remark}
    Many of the known special functions possess asymptotic expansions of the form proposed in the Assumption \ref{assu1}, though not all; see, for instance, the P\"oschl--Teller potential studies in \cite{FS25}.
\end{remark}

Even with this assumption, a complete analysis of the analytic continuation of the spectral $\z$-function is not feasible. However, we can study the meromorphic structure of the $\z$-function for $\Re(s)\geq 0$, including the important value $\zeta'(0;T_i)$, since, as it was recently shown in \cite[Sec. 3]{FPS25a}, it only depends on leading term of the small, and certain terms of the large, asymptotic expansion of the characteristic functions. 

From the previous section, it is not difficult to realize that $ F_{[M]}(z)$, $ F_{\alpha,[M]}(z)$, and $ F_{\alpha,\beta,[M]}(z)$ have all the same form. In particular, in the case of the single interface considered here, the three characteristic functions above can be written collectively as
\begin{align}
    F(z)=\wti f_2(z,R)[m_{21}\wti f_1(z,R)+m_{22}\wti f'_1(z,R)]-\wti f'_2(z,R)[m_{11}\wti f_1(z,R)+m_{12}\wti f'_1(z,R) ],
\end{align}
where $f_1(z,x)$ solves $\tau_1f(z,x)=zf(z,x)$ on $I_1$, $f_2(z,x)$ solves $\tau_2f(z,x)=zf(z,x)$ on $I_2$, and $m_{ij}$ are the entries of the $SL(2,\R)$ matrix $M$ describing the interface at $R$. By utilizing the asymptotic expansions assumed in \eqref{93} and \eqref{94} we have, for $\epsilon>0$ small enough,
\begin{align}\label{96}
m_{ij}\wti f_1(z,R)+m_{lk}\wti f'_1(z,R)= \Omega_1(m_{ij},m_{lk})z^{-\mu(m_{ij},m_{lk})}\exp\left\{\sum_{j=1}^{L_1}d_{j}^{(1)} z^{\beta_{j}^{(1)}}\right\}\left[1+O(z^{-\epsilon})\right],
\end{align}
where
\begin{align}
    \mu(m_{ij},m_{lk})=\begin{cases}
    \chi_1, & m_{lk}=0,\\
     \vartheta_1, &  m_{ij}=0,\\
     \textrm{min}\{\chi_1,\vartheta_1\}, & \textrm{otherwise},
    \end{cases}
    \end{align}
    and
    \begin{align} \Omega_1(m_{ij},0)=m_{ij} C_1A_{0}^{(1)},\ \, \Omega_1(0,m_{lk})=m_{lk}\bar{C}_1B_{0}^{(1)},\ \,
\Omega_1(m_{ij},m_{lk})=\begin{cases}
       m_{ij}C_1A_{0}^{(1)} , & \vartheta_1>\chi_1,\\
        m_{lk}\bar{C}_1B_{0}^{(1)} , & \vartheta_1<\chi_1,\\
      m_{ij}C_1A_{0}^{(1)}+ m_{lk}\bar{C}_1B_{0}^{(1)} , & \vartheta_1=\chi_1.
   \end{cases}
\end{align}
By using \eqref{93}, \eqref{94} and \eqref{96} we find, for $\iota>0$ small enough,
\begin{align}\label{99}
F(z)=\Xi(M)z^{-\mu(M)}\exp\left\{\sum_{j=1}^{L_1}d_{j}^{(1)} z^{\beta_{j}^{(1)}}+\sum_{j=1}^{L_2}d_{j}^{(2)} z^{\beta_{j}^{(2)}}\right\}\left[1+O(z^{-\iota})\right],    
\end{align}
with
\begin{align}\label{100}
    \mu(M)&=\textrm{min}\{\chi_2+\mu(m_{21},m_{22}),\vartheta_2+\mu(m_{11},m_{12})\},\nonumber\\
    \Xi(M)&=\begin{cases}
        C_2A_{0}^{(2)}\Omega_1(m_{21},m_{22}) , & \vartheta_2+\mu(m_{11},m_{12})>\chi_2+\mu(m_{21},m_{22}),\\
        -\bar{C}_2B_{0}^{(2)}\Omega_1(m_{11},m_{12}) , & \vartheta_2+\mu(m_{11},m_{12})<\chi_2+\mu(m_{21},m_{22}),\\
        C_2A_{0}^{(2)}\Omega_1(m_{21},m_{22})-\bar{C}_2B_{0}^{(2)}\Omega_1(m_{11},m_{12}) , & \vartheta_2+\mu(m_{11},m_{12})=\chi_2+\mu(m_{21},m_{22}).
    \end{cases}
\end{align}
To analytically continue the spectral $\z$-function \eqref{3.2} to a neighborhood of $s=0$ we subtract, and then add, from the integrand in the first integral of \eqref{3.2}, the asymptotic expansion of $\ln \left[z^{-m_{0}}F(z)\right]$ which, according to \eqref{99} reads (where we have chosen the analytic branch of the logarithm for large-$z$ about the branch cut as it avoids the zeros of $F(z)$)
\begin{align}
   \ln\left[ z^{-m_{0}}F(z)\right]=\ln\Xi(M)-[\mu(M)+m_0]\ln z+\sum_{j=1}^{L_1}d_{j}^{(1)} z^{\beta_{j}^{(1)}}+\sum_{j=1}^{L_2}d_{j}^{(2)} z^{\beta_{j}^{(2)}}+O(z^{-\iota}). 
\end{align}
Performing this operation leads to the expression
\begin{align}\label{102}
    \zeta(s;T_i)=Z(s)+Q(s)+{\cal L}(s),
\end{align}
where from Theorem \ref{Thrm3.4}
\begin{align}\label{103}
 Z(s)=e^{is(\pi-\Psi)}\frac{\sin(\pi s)}{\pi}\int_{\delta}^{\infty}dt\,t^{-s}\frac{d}{dt}\Bigg\{\ln\left[\left(te^{i\Psi}\right)^{-m_{0}}F\left(te^{i\Psi}\right)\right] -\ln\Xi(M)\nonumber\\
 +[\mu(M)+m_0]\ln\left(te^{i\Psi}\right)-\sum_{j=1}^{L_1}d_{j}^{(1)} \left(te^{i\Psi}\right)^{\beta_{j}^{(1)}}-\sum_{j=1}^{L_2}d_{j}^{(2)} \left(te^{i\Psi}\right)^{\beta_{j}^{(2)}} \Bigg\}, 
\end{align}
which is holomorphic for $\Re(s)>-\iota$,
\begin{align}\label{104}
    Q(s)=-\frac{\delta^{-s}}{2\pi i}\int_{\Psi-2\pi}^{\Psi}d\theta\,e^{-si\theta}\frac{d}{d\theta}\ln\,F\left(\delta e^{i\theta}\right),
\end{align}
that is an entire function of $s\in\C$, and
\begin{align}\label{105}
    {\cal L}(s)=e^{is(\pi-\Psi)}\frac{\sin(\pi s)}{\pi}\delta^{-s}\left[-\frac{\mu(M)+m_0}{s}+\sum_{j=1}^{L_1}\frac{d_j^{(1)}\beta_j^{(1)}(\delta e^{i\Psi})^{\beta_j^{(1)}}}{s-\beta_j^{(1)}}+\sum_{j=1}^{L_2}\frac{d_j^{(2)}\beta_j^{(2)}(\delta e^{i\Psi})^{\beta_j^{(2)}}}{s-\beta_j^{(2)}}\right],
\end{align}
is a meromorphic function that contains the information about the poles of $\z(s;T_i)$ in the region $\Re(s)>-\iota$. 
It is clear from \eqref{105} that if $\beta_j^{(1)},\beta_j^{(2)}\notin\N$ then $\z(s;T_i)$ has simple poles at $s=\beta_j^{(1)}$ and $s=\beta_j^{(2)}$ with residue 
\begin{align}
    \textrm{Res}[\zeta(s;T_i);\, s=\beta_j^{(n)}]=e^{i\pi\beta_j^{(n)}}\frac{\sin(\pi\beta_j^{(n)})}{\pi}d_j^{(n)}\beta_{j}^{(n)},\quad n\in\{1,2\}.
\end{align}
Note that when $\beta^{(1)}_l=\beta_m^{(2)}=\beta\notin\N$, with $l\in\{1,\ldots,L_1\}$ and $m\in\{1,\ldots,L_2\}$, two poles merge into a single one at $s=\beta$ with residue
\begin{align}
    \textrm{Res}[\zeta(s;T_i);\, s=\beta]=e^{i\pi\beta}\frac{\sin(\pi\beta)}{\pi}\beta\left(d_l^{(1)}+d_m^{(2)}\right). 
\end{align}
If $\beta_j^{(n)}\in\N$ with $n\in\{1,2\}$, then the spectral $\z$-function is regular at $s=\beta_j^{(n)}$ due to the zero of the prefactor $\sin(\pi s)$, and has the value
\begin{align}
    \zeta(\beta_j^{(n)};T_i)=d_j^{(n)}\beta_j^{(n)},\quad \beta_j^{(n)}\in\N .
\end{align}
We would like to point out that under Assumption \ref{assu1} the spectral $\z$-function is regular at the point $s=0$ and
\begin{align}\label{108a}
    \zeta(0;T_i)=-\mu(M)-m_0.
\end{align}

The derivative at $s=0$ can be computed by using the explicit expressions \eqref{103}-\eqref{105}. From \eqref{103} we have
\begin{align}\label{109}
    Z'(0)&=\int_{\delta}^{\infty}dt\,\frac{d}{dt}\Bigg\{\ln\left[\left(te^{i\Psi}\right)^{-m_{0}}F\left(te^{i\Psi}\right)\right] -\ln\Xi(M)\nonumber\\
 &+[\mu(M)+m_0]\ln\left(te^{i\Psi}\right)-\sum_{j=1}^{L_1}d_{j}^{(1)} \left(te^{i\Psi}\right)^{\beta_{j}^{(1)}}-\sum_{j=1}^{L_2}d_{j}^{(2)} \left(te^{i\Psi}\right)^{\beta_{j}^{(2)}} \Bigg\}\\
    &\hspace{-.7cm}=-\ln\left[\left(\delta e^{i\Psi}\right)^{-m_{0}}F\left(\delta e^{i\Psi}\right)\right]+\ln\Xi(M)-[\mu(M)+m_0]\ln\left(\delta e^{i\Psi}\right)+\sum_{j=1}^{L_1}d_{j}^{(1)} \left(\delta e^{i\Psi}\right)^{\beta_{j}^{(1)}}+\sum_{j=1}^{L_2}d_{j}^{(2)} \left(\delta e^{i\Psi}\right)^{\beta_{j}^{(2)}}.\nonumber
\end{align}
For the function $Q(s)$, one can show (cf. \cite[Thm. 2.6]{FPS25a}) that for $\Re(s)<1$,
\begin{align}
-\frac{1}{2\pi i}\int_{C_\delta}dz\,z^{-s}\frac{d}{dz}\ln\left[z^{-m_{0}}F\left(z\right)\right]=e^{is(\pi-\Psi)}\frac{\sin(\pi s)}{\pi}\int_0^{\delta}dt\,t^{-s}\frac{d}{dt}\ln\left[\left(te^{i\Psi}\right)^{-m_0}F\left(te^{i\Psi}\right)\right],
\end{align}
which can be used to compute
\begin{align}\label{111}
    Q'(0)=\int_0^{\delta}dt\,\frac{d}{dt}\ln\left[\left(te^{i\Psi}\right)^{-m_0}F\left(te^{i\Psi}\right)\right]=\ln\left[\left(\delta e^{i\Psi}\right)^{-m_{0}}F\left(\delta e^{i\Psi}\right)\right]-\lim_{t\to0^{+}}\ln\left[\left(te^{i\Psi}\right)^{-m_0}F\left(te^{i\Psi}\right)\right].
\end{align}
Since $F(z)$ is an entire function, it has the small-$z$ asymptotic expansion 
\begin{align}
    F(z)=F_{m_0}z^{m_0}+O(z^{m_0+1}),\quad F_{m_0}\neq 0,
\end{align}
so \eqref{111} becomes
\begin{align}\label{113}
   Q'(0)= \ln\left[\left(\delta e^{i\Psi}\right)^{-m_{0}}F\left(\delta e^{i\Psi}\right)\right]-\ln F_{m_0}.
\end{align}
Lastly, from \eqref{105} we obtain
\begin{align}\label{114}
    {\cal L}'(0)=[\mu(M)+m_0][-i(\pi-\Psi)+\ln \delta]-\sum_{j=1}^{L_1}d_{j}^{(1)} \left(\delta e^{i\Psi}\right)^{\beta_{j}^{(1)}}-\sum_{j=1}^{L_2}d_{j}^{(2)} \left(\delta e^{i\Psi}\right)^{\beta_{j}^{(2)}}.
\end{align}
By adding \eqref{109}, \eqref{113}, and \eqref{114} we get
\begin{align}\label{116}
    \zeta'(0;T_i)=-i\pi[\mu(M)+m_0]+\ln\Xi(M)-\ln F_{m_0}. 
\end{align}

\medskip

We now consider a specific example that explicitly illustrates the interface setting and the dependence of $\z'(0;T_i)$ on the matrix $M$. In the interval $I_1=(0,R)$ we consider the ordinary Bessel differential expression while in the interval $I_{2}=(R,\infty)$ we have the Airy differential expression $\tau_2$, that is 
\begin{align}
    \tau_1=-\frac{d^2}{dx^2}+\left(\nu^{2}-\frac{1}{4}\right)x^{-2},\quad \nu\in[0,\infty),\quad\textrm{and}\quad \tau_2=-\frac{d^2}{dx^2}+x.
\end{align}
The point $x=0$ is limit circle for $\tau_1$ when $\nu\in[0,1)$ (it is regular for $\nu=1/2$), and is limit point when $\nu\in[1,\infty)$. The interface is placed at a regular point $R$ for both operators, and $x=\infty$ is a limit point for the Airy operator $\tau_2$. Solutions to $\tau_{1} u=zu$ are given for $\nu\geq 0$ by (cf.\ \cite[No.~2.162, p.~440]{Ka61})
\begin{align}\label{114a}
f_{1,\nu}(z,x)=\sqrt{x} J_{\nu}\big(\sqrt{z}\, x\big),\quad\textrm{and}\quad g_{1,\nu}(z,x)=
\sqrt{x} Y_{\nu}\big(\sqrt{z}\, x\big),
\end{align}
where $J_{\mu}(\dott), Y_{\mu}(\dott)$ are the standard Bessel functions of order $\mu \in \R$ 
(cf.\ \cite[Ch.~9]{AS72}). The principal and nonprincipal solutions $u_{0,\nu}(0,x)$ and $\hatt u_{0,\nu}(0,x)$ of $\tau_1 u=0$ at the point $x=0$ are
\begin{align}
   u_{0,\nu}(0,x)=x^{\frac{1}{2}+\nu} ,\quad\nu\in[0,\infty),\quad\textrm{and}\quad 
   \hatt u_{0,\nu}(0,x)=\begin{cases} (2\nu)^{-1} x^{\frac{1}{2}-\nu} , & \nu \in (0,\infty),   \\
\sqrt{x} \ln(1/x), & \nu =0.  \end{cases}
\end{align}
In the interval $I_2$ the solutions of the Airy equation $\tau_2 u=zu$ are
\begin{align}
    f_{2}(z,x)=\textrm{Ai}(x-z),\quad\textrm{and}\quad g_{2}(z,x)=\textrm{Bi}(x-z),
\end{align}
of which $\textrm{Ai}(x-z)$ is the only one in $L^{2}((R,\infty),dx)$. 

We now need to consider two separate cases. One in which $x=0$ is a limit circle (including regular) point and the other in which it is a limit point. We illustrate, first, the case in which $x=0$ is limit circle.  

To construct the characteristic function $F_{\alpha,[M]}(z)$ we need the Weyl solutions (cf. \eqref{72}) at $x=0$. The generalized boundary values of $f_{1,\nu}(z,x)$ and $g_{1,\nu}(z,x)$ at $x=0$ are for $\nu\in(0,1)$ (cf. \eqref{A6} and \eqref{A7})
\begin{align}
 \wti f_{1,\nu}(z,0)&=0,\quad \wti g_{1,\nu}(z,0)=-\frac{2^{\nu+1}}{\pi}\Gamma(\nu+1)z^{-\nu/2}, \label{117}  \\
 \wti f'_{1,\nu}(z,0)&=\frac{z^{\nu/2}}{2^{\nu}\Gamma(\nu+1)},\quad \wti g^{\prime}_{1,\nu}(z,0)=-\frac{\cos(\pi\nu)\Gamma(-\nu)}{2^{\nu}\pi}z^{\nu/2}. \label{118}
\end{align}
Note that for $\nu=1/2$, $x=0$ is a regular point and \eqref{117}-\eqref{118} reduce to the ordinary boundary values of \eqref{114a}.  
For $\nu=0$ we have, instead,
\begin{align}
 \wti f_{1,0}(z,0)&=0,\quad \wti g_{1,0}(z,0)=-\frac{2}{\pi},  \\
 \wti f'_{1,0}(z,0)&=1,\quad \wti g^{\prime}_{1,0}(z,0)=\frac{2}{\pi}\left[\ln\left(\frac{\sqrt{z}}{2}\right)+\gamma_E\right],
\end{align}
where $\gamma_E$ is the Euler-Mascheroni constant. 
By imposing the conditions \eqref{72} at $x=0$ we get the Weyl solution and its derivative in $I_1$
\begin{align}\label{125}
    \Psi_{\alpha}(z,x,\nu)&=\sqrt{x}\left[A_{\nu}(\alpha)z^{-\nu/2}J_{\nu}(\sqrt{z}\,x)+B_{\nu}(\alpha)z^{\nu/2}Y_{\nu}(\sqrt{z}\,x)\right],\nonumber\\
    \Psi'_{\alpha}(z,x,\nu)&=\frac{1}{2x} \Psi_{\alpha}(z,x,\nu)+\sqrt{x}\left[A_{\nu}(\alpha)z^{(1-\nu)/2}J'_{\nu}(\sqrt{z}\,x)+B_{\nu}(\alpha)z^{(1+\nu)/2}Y'_{\nu}(\sqrt{z}\,x)\right],
\end{align} 
where
\begin{align}
  A_{\nu}(\alpha)&=\begin{cases}\label{126}
      2^{\nu}\Gamma(\nu+1)\cos(\alpha)-\cot(\nu\pi)z^\nu 2^{-\nu-1}\pi(\Gamma(\nu+1))^{-1}\sin(\alpha),& \nu\in(0,1),\\
      \cos(\alpha)-\sin(\alpha)\left(\frac{1}{2}\ln z-\ln 2+\gamma_E\right),& \nu=0,
  \end{cases}\\
  B_{\nu}(\alpha)&=\begin{cases}\label{127}
      2^{-\nu-1}\pi(\Gamma(\nu+1))^{-1}\sin(\alpha),& \nu\in(0,1),\\
      \frac{\pi}{2}\sin(\alpha), & \nu=0.
  \end{cases}
\end{align}
The characteristic function is then given by
\begin{align}
    F_{\alpha,[M]}(z)=\textrm{Ai}(R-z)[m_{21}\Psi_{\alpha}(z,R,\nu)+m_{22}\Psi'_{\alpha}(z,R,\nu)]-\textrm{Ai}'(R-z)[m_{11}\Psi_{\alpha}(z,R,\nu)+m_{12}\Psi'_{\alpha}(z,R,\nu)],
\end{align}
with $\Psi_{\alpha}(z,R,\nu)$ and $\Psi'_{\alpha}(z,R,\nu)$ denoting the value of the Weyl solution and its derivative in \eqref{125} at the regular point $R$. To evaluate the leading terms of the large-$z$ asymptotic expansion valid for $\textrm{Im}(z)>0$, we rewrite the Weyl solution and the Airy function in terms of Hankel functions as follows (cf. \cite[Eq. 10.4.4] {DLMF}) 
\begin{align}\label{129}
    &\Psi_{\alpha}(z,x,\nu)=\frac{\sqrt{x}}{2}\left[\left(\mathcal{A}_{\nu}(\alpha)z^{-\nu/2}-i\mathcal{B}_{\nu}(\alpha)z^{\nu/2}\right)H_{\nu}^{(1)}(\sqrt{z}\,x)+\left(\mathcal{A}_{\nu}(\alpha)z^{-\nu/2}+i\mathcal{B}_{\nu}(\alpha)z^{\nu/2}\right)H_{\nu}^{(2)}(\sqrt{z}\,x)\right],\nonumber\\
    &\Psi'_{\alpha}(z,x,\nu)=\frac{1}{2x} \Psi_{\alpha}(z,x,\nu)\nonumber\\
    &+\frac{\sqrt{x}}{2}\left[\left(\mathcal{A}_{\nu}(\alpha)z^{(1-\nu)/2}-i\mathcal{B}_{\nu}(\alpha)z^{(1+\nu)/2}\right)H_{\nu}^{(1)\prime}(\sqrt{z}\,x)+\left(\mathcal{A}_{\nu}(\alpha)z^{(1-\nu)/2}+i\mathcal{B}_{\nu}(\alpha)z^{(1+\nu)/2}\right)H_{\nu}^{(2)\prime}(\sqrt{z}\,x)\right],
\end{align}
where
\begin{align}
\mathcal{A}_\nu(\alpha)&=A_\nu(\alpha)+\cot(\nu\pi)z^\nu 2^{-\nu-1}\pi(\Gamma(\nu+1))^{-1}\sin(\alpha),\\ \mathcal{B}_\nu(\alpha)&=B_\nu(\alpha)+i\cot(\nu\pi) 2^{-\nu-1}\pi(\Gamma(\nu+1))^{-1}\sin(\alpha),
\end{align}
and \cite[Eqs. 9.6.6 and 9.6.7] {DLMF}
\begin{align}
 \Ai(-z)=\frac{1}{2}\sqrt{\frac{z}{3}}\left(e^{-i\pi/6}H_{1/3}^{(2)}\left(\frac{2}{3}z^{3/2}\right)+e^{i\pi/6}H_{1/3}^{(1)}\left(\frac{2}{3}z^{3/2}\right)\right),
\end{align}
\begin{align}
    \textrm{Ai}'\left(-z\right)=\frac{z}{2\sqrt{3}}\left(e^{-\pi i/6}{H^{(1)}_
{2/3}}\left(\frac{2}{3}z^{3/2}\right)+e^{\pi i/6}{H^{(2)}_{2/3}}\left(\frac{2}{3}z^{3/2}\right)\right).
\end{align}

The large-$z$ asymptotic expansion of $\Psi_{\alpha}(z,R,\nu)$ and $\Psi'_{\alpha}(z,R,\nu)$ depend on whether or not $\alpha=0$ since from equation \eqref{126} we have that $\mathcal{B}_{\nu}(0)=0$ when $\alpha=0$. We will, therefore, distinguish between the cases $\alpha=0$ and $\alpha\neq0$.

\paragraph{Case $\alpha\neq 0$.} 
According to \cite[Eqs. 10.2.5 and 10.2.6]{DLMF} the Hankel function $H_{1/3}^{(2)}\left((2/3)z^{3/2}\right)$ is dominant compared to $H_{1/3}^{(1)}\left((2/3)z^{3/2}\right)$ for $z\to\infty$ with $\textrm{Im}(z)>0$. This implies that by 
using the expression \cite[Eq. 10.17.6]{DLMF},
\begin{align}
    H_{\nu}^{(2)}(w)=\left(\frac{2}{\pi w}\right)^{1/2}e^{-i\left(w-\frac{\nu\pi}{2}-\frac{\pi}{4}\right)}\left(1+O(w^{-1})\right),\quad H_{\nu}^{(2)\prime}(w)=-i\left(\frac{2}{\pi w}\right)^{1/2}e^{-i\left(w-\frac{\nu\pi}{2}-\frac{\pi}{4}\right)}\left(1+O(w^{-1})\right),\nonumber
\end{align}
where the expansion of the derivative is obtained by using \cite[Eq. 10.6.1]{DLMF}, that we have, for large values of $z$ with $\textrm{Im}(z^{1/2})>0$, 
\begin{align}\label{132}
    \Psi_{\alpha}(z,R,\nu)&=i\frac{\mathcal{B}_{\nu}(\alpha)}{\sqrt{2\pi}}z^{\frac{\nu}{2}-\frac{1}{4}}\exp\left\{-i\sqrt{z}R+\frac{i\pi\nu}{2}+\frac{i\pi}{4}\right\}\left(1+O(z^{-\nu})\right),\nonumber\\
    \Psi'_{\alpha}(z,R,\nu)&=\frac{\mathcal{B}_{\nu}(\alpha)}{\sqrt{2\pi}}z^{\frac{\nu}{2}+\frac{1}{4}}\exp\left\{-i\sqrt{z}R+\frac{i\pi\nu}{2}+\frac{i\pi}{4}\right\}\left(1+O(z^{-\nu})\right). 
\end{align}

For the Airy functions we obtain 
\begin{align}\label{133}
    \textrm{Ai}(R-z)&=\frac{1}{2\sqrt{\pi}}(z-R)^{-1/4}\exp\left\{-i\frac{2}{3}(z-R)^{3/2}+\frac{i\pi}{4}\right\}\left(1+O\left((z-R)^{-3/2}\right)\right)\nonumber\\
    &=\frac{1}{2\sqrt{\pi}}z^{-1/4}\exp\left\{-i\frac{2}{3}z^{3/2}+i R\,\sqrt{z}+\frac{i\pi}{4}\right\}\left(1+O\left(z^{-1/2}\right)\right),\nonumber\\
     \textrm{Ai}'(R-z)&=\frac{i}{2\sqrt{\pi}}(z-R)^{1/4}\exp\left\{-i\frac{2}{3}(z-R)^{3/2}+\frac{i\pi}{4}\right\}\left(1+O\left((z-R)^{-3/2}\right)\right) \nonumber\\
     &=\frac{1}{2\sqrt{\pi}}z^{1/4}\exp\left\{-i\frac{2}{3}z^{3/2}+i R\,\sqrt{z}+\frac{i\pi}{4}\right\}\left(1+O\left(z^{-1/2}\right)\right).
\end{align}

By comparing \eqref{132} and \eqref{133} with \eqref{93} and \eqref{94} we have $A_{0}^{(1)}=A_{0}^{(2)}=B_{0}^{(1)}=B_{0}^{(2)}=1$,
\begin{align}
    C_1&=\frac{iB_{\nu}(\alpha)}{\sqrt{2\pi}}\exp\left\{\frac{i\pi\nu}{2}+\frac{i\pi}{4}\right\},\quad C_2=\frac{1}{2\sqrt{\pi}}\exp\left\{\frac{i\pi}{4}\right\},\quad \bar{C}_1=-iC_1,\quad\bar{C}_{2}=C_2,\nonumber\\
    \chi_1&=\frac{1}{4}-\frac{\nu}{2},\quad d_{1}^{(1)}=-iR,\quad \beta_{1}^{(1)}=\frac{1}{2},\quad \vartheta_1=\chi_1-\frac{1}{2},
    \nonumber\\
   \chi_2&=\frac{1}{4},\quad d_{1}^{(2)}=-\frac{2i}{3},\quad d_{2}^{(2)}=iR,\quad \beta_{1}^{(2)}=\frac{3}{2},\quad \beta_{2}^{(2)}=\frac{1}{2},\quad \vartheta_2=-\chi_2.
\end{align}

From these expressions we have 
\begin{align}
    \mu(m_{21},m_{22})=\begin{cases}
    \frac{1}{4}-\frac{\nu}{2}, & m_{22}=0, \\
     -\frac{1}{4}-\frac{\nu}{2}, &  m_{22}\neq 0,
    \end{cases}\quad \mu(m_{11},m_{12})=\begin{cases}
    \frac{1}{4}-\frac{\nu}{2}, & m_{12}=0, \\
     -\frac{1}{4}-\frac{\nu}{2}, &  m_{12}\neq 0,
    \end{cases}
    \end{align}
which, by using \eqref{100}, allows us to obtain
\begin{align}
    \mu(M)=\begin{cases}
   -\frac{\nu}{2}, & m_{12}=0, \\
     -\frac{1}{2}-\frac{\nu}{2}, &  m_{12}\neq 0.
    \end{cases}
\end{align}
Moreover, by noting that $\vartheta_i=\chi_i-1/2$ for $i\in\{1,2\}$, we have
\begin{align}
    \Omega_1(m_{21},m_{22})&=\begin{cases}
       m_{21}i\frac{\mathcal{B}_{\nu}(\alpha)}{\sqrt{2\pi}}\exp\left\{\frac{i\pi\nu}{2}+\frac{i\pi}{4}\right\} , & m_{22}=0,\\
        m_{22}\frac{\mathcal{B}_{\nu}(\alpha)}{\sqrt{2\pi}}\exp\left\{\frac{i\pi\nu}{2}+\frac{i\pi}{4}\right\}, & m_{22}\neq 0, 
   \end{cases}\nonumber\\
   \Omega_1(m_{11},m_{12})&=\begin{cases}
       m_{11}i\frac{\mathcal{B}_{\nu}(\alpha)}{\sqrt{2\pi}}\exp\left\{\frac{i\pi\nu}{2}+\frac{i\pi}{4}\right\} , & m_{12}=0,\\
        m_{12}\frac{\mathcal{B}_{\nu}(\alpha)}{\sqrt{2\pi}}\exp\left\{\frac{i\pi\nu}{2}+\frac{i\pi}{4}\right\}, & m_{12}\neq 0, 
   \end{cases}
\end{align}
which can be used to obtain, from \eqref{100}, the relation
\begin{align}
\Xi(M)=\frac{\mathcal{B}_{\nu}(\alpha)}{2\sqrt{2}\pi}\exp\left\{\frac{i\pi\nu}{2}\right\}\times\begin{cases}
        -im_{12} , & m_{12}\neq 0,\\
        m_{11}+im_{22}, & m_{12}=0.
    \end{cases}
\end{align}

To complete the evaluation of $\zeta'(0;T_{\alpha,[M]})$, we need the leading term of the small-$z$ expansion of the characteristic function. By using the formulas in \eqref{125} one obtains, for $\nu\in(0,1)$,
\begin{align}
    \Psi_{\alpha}(z,R,\nu)&=R^{\nu+1/2}\cos(\alpha)-\frac{R^{1/2-\nu}}{2\nu}\sin(\alpha)+O(z),\nonumber\\
     \Psi'_{\alpha}(z,R,\nu)&=R^{\nu-1/2}\left(\frac{1}{2}+\nu\right)\cos(\alpha)-\frac{R^{-1/2-\nu}}{2\nu}\left(\frac{1}{2}-\nu\right)\sin(\alpha)+O(z),
\end{align}
while for $\nu=0$ we have
\begin{align}
    \Psi_{\alpha}(z,R,0)&=R^{1/2}\left[\cos(\alpha)+\ln R\,\sin(\alpha)\right]+O(z),\nonumber\\
     \Psi'_{\alpha}(z,R,0)&=\frac{1}{2}R^{-1/2}\left[\cos(\alpha)+\left(\ln R+2\right)\sin(\alpha)\right]+O(z). 
\end{align}
These results allow us to conclude that when  $\nu\in(0,1)$
\begin{align}\label{141}
    F_{\alpha,[M]}(0)=\sqrt{R}\,\textrm{Ai}(R)\Big[R^{\nu}\cos(\alpha)\left(m_{21}+\left(\nu+\frac{1}{2}\right)\frac{m_{22}}{R}\right)-\frac{R^{-\nu}}{2\nu}\sin(\alpha)\left(m_{21}+\left(\frac{1}{2}-\nu\right)\frac{m_{22}}{R}\right)\Big]\nonumber\\
    -\sqrt{R}\,\textrm{Ai}'(R)\Big[R^{\nu}\cos(\alpha)\left(m_{11}+\left(\nu+\frac{1}{2}\right)\frac{m_{12}}{R}\right)-\frac{R^{-\nu}}{2\nu}\sin(\alpha)\left(m_{11}+\left(\frac{1}{2}-\nu\right)\frac{m_{12}}{R}\right)\Big],
\end{align}
and for $\nu=0$,
\begin{align}\label{142}
    F_{\alpha,[M]}(0)=\sqrt{R}\,\textrm{Ai}(R)\Big[\cos(\alpha)\left(m_{21}+\frac{m_{22}}{2R}\right)+\sin(\alpha)\left(m_{21}\ln R+\frac{m_{22}}{2R}(\ln R+2)\right)\Big]\nonumber\\
    -\sqrt{R}\,\textrm{Ai}'(R)\Big[\cos(\alpha)\left(m_{11}+\frac{m_{12}}{2R}\right)+\sin(\alpha)\left(m_{11}\ln R+\frac{m_{12}}{2R}(\ln R+2)\right)\Big].
\end{align}
For the sake of simplicity, we assume that \eqref{141} and \eqref{142} are not vanishing so that $m_0=0$. If, for some particular values of the parameters that enter those expressions, they vanish, then one would have to compute the next non-vanishing term of the small-$z$ expansion of $F_{\alpha,[M]}(z)$. 

The results obtained above allow us to compute, by using \eqref{108a}, 
\begin{align}
    \zeta(0;T_{\alpha,[M]})=\begin{cases}
   \frac{\nu}{2}, & m_{12}=0, \\
     \frac{1}{2}+\frac{\nu}{2}, &  m_{12}\neq 0,
    \end{cases}
\end{align}
and, by utilizing \eqref{116},
\begin{align}
    \zeta'(0;T_{\alpha,[M]})=i\pi\nu+\ln\left(\frac{\mathcal{B}_{\nu}(\alpha)}{2\sqrt{2}\pi}\right)-\ln F_{\alpha,[M]}(0)+\begin{cases}
        \ln(m_{11}+im_{22}), & m_{12}=0, \\
        \ln m_{12}, &  m_{12}\neq 0,
    \end{cases}
\end{align}
with $F_{\alpha,[M]}(0)$ given in \eqref{141} when $\nu\in(0,1)$, and in \eqref{142} when $\nu=0$. 
\paragraph{Case $\alpha= 0$.}
When $\alpha=0$, the large-$z$ asymptotic expansion of $\Psi_{0}(z,R,\nu)$ and $\Psi'_{0}(z,R,\nu)$ read
\begin{align}\label{132a}
    \Psi_{0}(z,R,\nu)&=\frac{A_{\nu}(0)}{\sqrt{2\pi}}z^{-\frac{\nu}{2}-\frac{1}{4}}\exp\left\{-i\sqrt{z}R+\frac{i\pi\nu}{2}+\frac{i\pi}{4}\right\}\left(1+O(z^{-1/2})\right),\nonumber\\
    \Psi'_{0}(z,R,\nu)&=-i\frac{A_{\nu}(0)}{\sqrt{2\pi}}z^{-\frac{\nu}{2}+\frac{1}{4}}\exp\left\{-i\sqrt{z}R+\frac{i\pi\nu}{2}+\frac{i\pi}{4}\right\}\left(1+O(z^{-1/2})\right).
\end{align}
By comparing \eqref{132a} and \eqref{133} with \eqref{93} and \eqref{94} we find that the coefficients of the general asymptotic expansions in \eqref{93} and \eqref{94} are the same as in the case $\alpha\neq0$ with the following exceptions:
\begin{align}
    C_1=\frac{A_{\nu}(0)}{\sqrt{2\pi}}\exp\left\{\frac{i\pi\nu}{2}+\frac{i\pi}{4}\right\},\quad
    \chi_1=\frac{\nu}{2}+\frac{1}{4},\quad\textrm{and}\quad \vartheta_1=\chi_1-\frac{1}{2}=\frac{\nu}{2}-\frac{1}{4}.
\end{align}
Consequently we have
\begin{align}
    \mu(m_{21},m_{22})=\begin{cases}
    \frac{1}{4}+\frac{\nu}{2}, &  m_{22}=0, \\
     -\frac{1}{4}+\frac{\nu}{2}, &  m_{22}\neq 0,
    \end{cases}\quad \mu(m_{11},m_{12})=\begin{cases}
    \frac{1}{4}+\frac{\nu}{2}, &  m_{12}=0, \\
     -\frac{1}{4}+\frac{\nu}{2}, &  m_{12}\neq 0,
    \end{cases}
    \end{align}
from which we can write 
\begin{align}
    \mu(M)=\begin{cases}
   \frac{\nu}{2} , &  m_{12}=0, \\
     -\frac{1}{2}+\frac{\nu}{2}, &  m_{12}\neq 0.
    \end{cases}
\end{align}
For the coefficient $\Xi(M)$ we find
\begin{align}
\Xi(M)=\frac{A_{\nu}(0)}{2\sqrt{2}\pi}\exp\left\{\frac{i\pi\nu}{2}\right\}\times\begin{cases}
        -m_{12} , & m_{12}\neq 0,\\
        m_{22}-im_{11}, & m_{12}=0.
    \end{cases}
\end{align}
Equation \eqref{116} allows us to obtain
\begin{align}\label{148}
    \zeta'(0;T_{0,[M]})=-\frac{i\pi}{2}+\ln\left(\frac{A_{\nu}(0)}{2\sqrt{2}\pi}\right)-\ln F_{0,[M]}(0)+\begin{cases}
        \ln(m_{11}+im_{22}), &  m_{12}=0, \\
        \ln m_{12}, &  m_{12}\neq 0,
    \end{cases}
\end{align}
where, $F_{0,[M]}(0)$ is computed by setting $\alpha=0$ in  \eqref{141} for $\nu\in(0,1)$ and in \eqref{142} for $\nu=0$.

We would like, at this point, to address the case in which $x=0$ is a limit point. In this instance, the normalized principal solution of $\tau_1 u=zu$ on $I_1$ (namely one that behaves like $x^{\nu+1/2}$ as $x\to0^{+}$) is
\begin{align}
    \Psi(z,x,\nu)=2^{\nu}\Gamma(\nu+1)z^{-\nu/2}\sqrt{x}J_{\nu}(\sqrt{z}\,x),\quad \nu\in[1,\infty),
\end{align}
which coincides with the solution $\Psi_{\alpha}(z,x,\nu)$ in \eqref{125} when $\alpha=0$. This means that the evaluation of $\zeta'(0;T_{0,[M]})$ when $x=0$ is limit point, follows exactly the one detailed above when $\alpha=0$ and $\nu\neq0$. In particular, $\zeta'(0;T_{0,[M]})$ is given by \eqref{148} where $\nu\in[1,\infty)$ and $A_{\nu}(0)=2^{\nu}\Gamma(\nu+1)$.  

\subsubsection{Two generalized point potentials}

We consider the case of two generalized point potentials located at $R_1$ and $R_2$ within a finite interval $I=[0,L]$. For simplicity, we will assume that the differential expression in each of the three sub-intervals is $\tau=-d^{2}/dx^{2}$. For this choice, the endpoints of the interval, as well as the points $R_1$ and $R_2$ are regular and, therefore, require ordinary boundary conditions. Furthermore, we focus on the case in which the endpoints of the interval are separated. In the first interval, $I_1=[0,R_1)$, the Weyl solution at $x=0$ is 
\begin{align}\label{150}
    \Psi_{\alpha}(z,x)=\cos(\alpha)\frac{\sin\left(\sqrt{z}\,x\right)}{\sqrt{z}}-\sin(\alpha)\cos\left(\sqrt{z}\,x\right).
\end{align}
In the second interval, $I_2=(R_1,R_2)$, the fundamental system of solutions normalized at $R_1$ is
\begin{align}
    \phi_2(z,x)=\frac{\sin\left(\sqrt{z}\,(x-R_1)\right)}{\sqrt{z}},\quad \theta_2(z,x)=\cos\left(\sqrt{z}\,(x-R_1)\right).
\end{align}
Lastly, in interval $I_3=(R_2,L]$ we have the Weyl solution at $x=L$ 
\begin{align}\label{152}
 \Psi_{\beta}(z,x)=-\cos(\beta)\frac{\sin\left(\sqrt{z}(L-x)\right)}{\sqrt{z}}+\sin(\beta)\cos\left(\sqrt{z}(L-x)\right).   
\end{align}
The $SL(2,\R)$ matrices representing the point potential at $R_1$ and $R_2$ are $M_1=[\nu_{ij}]$ and $M_2=[\mu_{ij}]$, respectively. 
According to \eqref{70}, the transfer matrix ${\cal T}_2$ for the two generalized point potentials has the form
\begin{align}\label{153}
     {\cal T}_2&=M_2\Phi_{_{2}}(R_{2})M_1=M_2\begin{pmatrix}
         \cos\left(\sqrt{z}\,R\right) & z^{-1/2}\sin\left(\sqrt{z}\,R\right)\\
         -z^{1/2}\sin\left(\sqrt{z}\,R\right) & \cos\left(\sqrt{z}\,R\right)
     \end{pmatrix}M_1 \nonumber\\
&=M_2M_1\cos\left(\sqrt{z}\,R\right)+M_2\Lambda(z)M_1\frac{\sin\left(\sqrt{z}\,R\right)}{\sqrt{z}},\quad \Lambda(z)=\begin{pmatrix}
    0&1\\
    -z&0
\end{pmatrix},\quad R=R_2-R_1.
\end{align}
By introducing the vectors 
\begin{align}\label{154}
    v_{\alpha}(z)=\begin{pmatrix}
        \Psi_\alpha(z,R_1)\\
        \Psi'_\alpha(z,R_1)
    \end{pmatrix},\quad\textrm{and}\quad u_{\beta}(z)=\begin{pmatrix}
        \Psi_\beta(z,R_2)\\
        \Psi'_\beta(z,R_2)
    \end{pmatrix},
\end{align}
the characteristic function in \eqref{73} can be expressed as
\begin{align}\label{155}
    F_{\alpha,\beta,[M]}(z)=u_{\beta}^{T}(z)J{\cal T}_2v_{\alpha}(z),\quad\textrm{with}\quad J=\begin{pmatrix}
        0&1\\-1&0
    \end{pmatrix}.
\end{align}

In order to compute $\zeta'(0;T_{\alpha,\beta,[M]})$ for the two generalized point potentials, we need the large-$z$ asymptotic expansion of $\ln F_{\alpha,\beta,[M]}(z)$ for which we will assume $m_0=0$ throughout. For $\textrm{Im}(\sqrt{z})>0$ we can write
\begin{align}\label{156}
   F_{\alpha,\beta,[M]}(z)=\frac{1}{2}e^{-i\sqrt{z}R} u_{\beta}^{T}(z)J\left[M_2M_1+\frac{i}{\sqrt{z}}M_2\Lambda(z)M_1\right]v_{\alpha}(z)\left(1+O\left(e^{2i\sqrt{z}R}\right)\right),
\end{align}
which is obtained by using the large-$z$ expansion of ${\cal T}_2$ in \eqref{153}. The matrix product in the square parentheses can be expressed as
\begin{align}
   M_2M_1+\frac{i}{\sqrt{z}}M_2\Lambda(z)M_1=N+\frac{i}{\sqrt{z}}Q-i\sqrt{z}P ,
\end{align}
where the matrices $N$, $Q$, and $P$ have entries
\begin{align}
    n_{ij}=\sum_{k=1}^{2}\mu_{ik}\nu_{kj},\quad q_{ij}=\mu_{i1}\nu_{2j},\quad\textrm{and}\quad p_{ij}=\mu_{i2}\nu_{1j}.
\end{align}

Next, we need the large-$z$ asymptotic expansion of the vectors $v_\alpha(z)$ and $u_{\beta}(z)$. By using the Weyl solutions \eqref{150} and \eqref{152} in \eqref{154} we obtain, for $\alpha\neq0$ and $\beta\neq 0$,
\begin{align}\label{159}
    v_{\alpha}(z)&=-\frac{\sin(\alpha)}{2}e^{-i\sqrt{z}R_1}\begin{pmatrix}
        1\\-i\sqrt{z}
    \end{pmatrix}\left(1+O\left(z^{-1/2}\right)\right),\nonumber\\
    u_{\beta}(z)&=\frac{\sin(\beta)}{2}e^{-i\sqrt{z}(L-R_2)}\begin{pmatrix}
        1\\i\sqrt{z}
    \end{pmatrix}\left(1+O\left(z^{-1/2}\right)\right),
\end{align}
and
\begin{align}\label{160}
   v_{0}(z)&=\frac{i}{2\sqrt{z}}e^{-i\sqrt{z}R_1}\begin{pmatrix}
        1\\-i\sqrt{z}
    \end{pmatrix}\left(1+O\left(e^{2i\sqrt{z}R_1}\right)\right),\nonumber\\
    u_{0}(z)&=-\frac{i}{2\sqrt{z}}e^{-i\sqrt{z}(L-R_2)}\begin{pmatrix}
        1\\i\sqrt{z}
    \end{pmatrix}\left(1+O\left(e^{2i\sqrt{z}(L-R_2)}\right)\right) .
\end{align}
By introducing the function $G_{\Omega}(z)$ associated to a $2\times 2$ matrix $\Omega$ defined as
\begin{align}\label{161}
G_{\Omega}(z)= \begin{pmatrix}
        1 &i\sqrt{z}
    \end{pmatrix}J\Omega \begin{pmatrix}
        1\\-i\sqrt{z}
    \end{pmatrix}=\Omega_{21}-i\sqrt{z}\,\textrm{Tr}\Omega-z\Omega_{12}, 
\end{align}
we have
\begin{align}\label{162}
    F_{\alpha,\beta,[M]}(z)=\frac{1}{8}S_{\alpha,\beta}z^{-\kappa_{\alpha,\beta}/2}e^{-i\sqrt{z}L}\left[G_{N}(z)+\frac{i}{\sqrt{z}}G_{Q}(z)-i\sqrt{z}G_{P}(z)\right]\left(1+O\left(\epsilon_{\alpha,\beta}(z)\right)\right),
\end{align}
where 
\begin{align}\label{163}
    S_{\alpha,\beta}=\begin{cases}
        -\sin(\alpha)\sin(\beta),& \alpha\neq0,\,\beta\neq 0,\\
        i\sin(\beta),& \alpha=0,\,\beta\neq 0,\\
        i\sin(\alpha),& \alpha\neq0,\,\beta= 0,\\
        1,& \alpha=\beta=0,
    \end{cases}\quad \kappa_{\alpha,\beta}=
\begin{cases}
0, & \alpha\neq0,\ \beta\neq0,\\
1, & \alpha=0,\ \beta\neq0\ \text{or}\ \alpha\neq0,\ \beta=0,\\
2, & \alpha=\beta=0,
\end{cases}
\end{align}
and the error is
\begin{align}
\epsilon_{\alpha,\beta}(z)=
\begin{cases}
z^{-1/2}, & \alpha\neq0,\ \beta\neq0,\\
e^{2i\sqrt{z}\,R_1}, & \alpha=0,\ \beta\neq0,\\
e^{2i\sqrt{z}\,(L-R_2)}, & \alpha\neq0,\ \beta=0,\\
e^{2i\sqrt{z}\,\eta}, & \alpha=\beta=0,
\end{cases} \quad   \textrm{with}\quad \eta=\min\{R_1,\;R,\;L-R_2\}.
\end{align}
Before deriving the large-$z$ asymptotic expansion of $\ln F_{\alpha,\beta,[M]}(z)$, it is convenient to write the function $G_{N}(z)+iz^{-1/2}G_{Q}(z)-iz^{1/2}G_{P}(z)$ as a combination of terms proportional to powers of $\sqrt{z}$, that is
\begin{align}\label{165}
    G_{N}(z)+\frac{i}{\sqrt{z}}G_{Q}(z)-i\sqrt{z}G_{P}(z)=\frac{\rho_{-1/2}}{\sqrt{z}}+\rho_0+\sqrt{z}\rho_{1/2}+z\rho_1+z^{3/2}\rho_{3/2},
\end{align}
with
\begin{align}
    \rho_{-1/2}=iq_{21},\quad \rho_0=n_{21}+\textrm{Tr}Q,\quad \rho_{1/2}=-i(q_{12}+p_{21}+\textrm{Tr}N),\quad \rho_1=-(n_{12}+\textrm{Tr}P),\quad \rho_{3/2}=ip_{12}.
\end{align}

By utilizing the above results, we can finally write (where we have assumed $m_0=0$)
\begin{align}\label{167}
    \ln F_{\alpha,\beta,[M]}(z)&=-i\sqrt{z}L+\ln \left(\frac{S_{\alpha,\beta}}{8}\right)-\frac{\kappa_{\alpha,\beta}}{2}\ln z\nonumber\\
    &+\ln\left[\frac{\rho_{-1/2}}{\sqrt{z}}+\rho_0+\sqrt{z}\rho_{1/2}+z\rho_1+z^{3/2}\rho_{3/2}\right]+O(\epsilon_{\alpha,\beta}(z)).
\end{align}
Let $\sigma=\textrm{max}\{\iota :\rho_{\iota}\neq0\}$, then we have the expansion
\begin{align}\label{168}
  \ln\left[\frac{\rho_{-1/2}}{\sqrt{z}}+\rho_0+\sqrt{z}\rho_{1/2}+z\rho_1+z^{3/2}\rho_{3/2}\right] =\sigma\ln z+\ln\rho_\sigma+O\left(z^{-(\sigma-\gamma)}\right), 
\end{align}
with $\gamma=\textrm{max}\{\iota<\sigma : \rho_\iota\neq0\}$.
By utilizing \eqref{168} in the expression \eqref{167} we obtain
\begin{align}\label{169}
  \ln F_{\alpha,\beta,[M]}(z)&=-i\sqrt{z}L+\ln \left(\frac{S_{\alpha,\beta}}{8}\right)+\ln\rho_\sigma +\left(\sigma-\frac{\kappa_{\alpha,\beta}}{2}\right)\ln z +O\left(\textrm{max}\{|\epsilon_{\alpha,\beta}(z)|,|z^{-(\sigma-\gamma)}|\}\right).
\end{align}

To complete the computation of $\zeta'(0;T_{\alpha,\beta,[M]})$ we need $F_{\alpha,\beta,[M]}(0)$ assuming $m_0=0$. By using \eqref{155} and noting that $M_2\Lambda(0)M_1=Q$ we get
\begin{align}\label{170}
    F_{\alpha,\beta,[M]}(0)=\lim_{z\to 0}u_{\beta}^{T}(z)J{\cal T}_2v_{\alpha}(z)=u_{\beta}^{T}(0)J\left(N+RQ\right)v_{\alpha}(0).
\end{align}
Recalling the expressions \eqref{150} and \eqref{152} we have, as $z\to 0$,
\begin{align}
 u_{\beta}(0)=\begin{pmatrix}
     -(L-R_2)\cos(\beta)+\sin(\beta)\\
     \cos(\beta)
 \end{pmatrix} ,\quad \textrm{and}\quad   v_{\alpha}(0)=\begin{pmatrix}
     R_1\cos(\alpha)-\sin(\alpha)\\
     \cos(\alpha)
 \end{pmatrix},
\end{align}
and, hence,
\begin{align}\label{172}
   F_{\alpha,\beta,[M]}(0)&=-\cos(\alpha)\cos(\beta)
\Bigl[
R_1(L-R_2)(n_{21}+Rq_{21})
+(L-R_2)(n_{22}+Rq_{22})
+R_1(n_{11}+Rq_{11})
\nonumber\\
&+n_{12}+Rq_{12}
\Bigr]
+\sin(\alpha)\cos(\beta)
\Bigl[
(L-R_2)(n_{21}+Rq_{21})
+n_{11}+Rq_{11}
\Bigr]
\nonumber\\
&+\cos(\alpha)\sin(\beta)
\Bigl[
R_1(n_{21}+Rq_{21})
+n_{22}+Rq_{22}
\Bigr]-\sin(\alpha)\sin(\beta)
(n_{21}+Rq_{21}),
\end{align}
where we once again assume this expression is nonvanishing.

The derivative of the spectral $\zeta$-function at $s=0$ is then
\begin{align}\label{173}
    \zeta'(0;T_{\alpha,\beta,[M]})=-i\pi\left(\sigma-\frac{\kappa_{\alpha,\beta}}{2}\right)+\ln \left(\frac{S_{\alpha,\beta}}{8}\right)+\ln\rho_\sigma-\ln F_{\alpha,\beta,[M]}(0). 
\end{align}

\begin{remark}
   It is important to make a comment about the limit of coalescing generalized point potentials. From equation \eqref{169} it is not difficult to see that the terms of the large-$z$ asymptotic expansion do not depend on the position of the generalized point potentials. This dependence can be found in the neglected part of the expansion. In the limit $R=R_2-R_1\to0$ the error term in the expansion \eqref{156} is no longer negligible, and, in fact, becomes of the same order as the leading term. This means that the subsequent large-$z$ expansion \eqref{169} is not valid in the limit $R\to0$. To obtain the correct large-$z$ asymptotic expansion of the characteristic function one needs to perform the $R\to 0$ limit first, and then expand for large $z$. 

   For $R\to0$, namely $R_1,R_2\to r$, we have, from \eqref{153}, that ${\cal T}_{2}\to M_2M_1$, as already mentioned at the end of Section \ref{Sec3.1}. This implies that the characteristic function becomes 
   \begin{align}
    F^{(R=0)}_{\alpha,\beta,[M]}(z)=u_{\beta}^{T}(z)JM_2M_1v_{\alpha}(z).
\end{align}
   By utilizing the expansions \eqref{159} and \eqref{160} we obtain
   \begin{align}
    F^{(R=0)}_{\alpha,\beta,[M]}(z)=\frac{1}{4}S_{\alpha,\beta}z^{-\kappa_{\alpha,\beta}/2}e^{-i\sqrt{z}L}G_{M_2M_1}(z)\left(1+O\left(\bar{\epsilon}_{\alpha,\beta}(z)\right)\right),  
   \end{align}
   where 
   \begin{align}
\bar{\epsilon}_{\alpha,\beta}(z)=
\begin{cases}
z^{-1/2}, & \alpha\neq0,\ \beta\neq0,\\
e^{2i\sqrt{z}\,r}, & \alpha=0,\ \beta\neq0,\\
e^{2i\sqrt{z}\,(L-r)}, & \alpha\neq0,\ \beta=0,\\
e^{2i\sqrt{z}\,\bar{\eta}}, & \alpha=\beta=0,
\end{cases} \quad   \textrm{with}\quad \bar{\eta}=\min\{r,\;L-r\}.
\end{align}
   These results lead to the expression 
\begin{align}
  \ln F^{(R=0)}_{\alpha,\beta,[M]}(z)&=-i\sqrt{z}L+\ln \left(\frac{S_{\alpha,\beta}}{4}\right)+\ln\rho_\sigma +\left(\sigma-\frac{\kappa_{\alpha,\beta}}{2}\right)\ln z +O\left(\textrm{max}\{|\bar{\epsilon}_{\alpha,\beta}(z)|,|z^{-(\sigma-\gamma)}|\}\right),
\end{align}
with $\rho_\sigma$ and $\sigma$ as defined in \eqref{168} and $\rho_0=(M_2M_1)_{21}$, $\rho_{1/2}=-i\textrm{Tr}(M_2M_1)$, and $\rho_{1}=-(M_2M_1)_{12}$.

The derivative of the spectral $\zeta$-function at $s=0$ is then
\begin{align}
   \zeta_{(R=0)}'(0;T_{\alpha,\beta,[M]})=-i\pi\left(\sigma-\frac{\kappa_{\alpha,\beta}}{2}\right)+\ln \left(\frac{S_{\alpha,\beta}}{4}\right)+\ln\rho_\sigma-\ln F^{(R=0)}_{\alpha,\beta,[M]}(0),  
\end{align}
where $F^{(R=0)}_{\alpha,\beta,[M]}(0)$ is obtained from \eqref{172} by simply taking the limits $R\to 0$ and $R_1,R_2\to r$ under the assumption that $m_0=0$.
\end{remark}

\medskip

As an explicit example, we consider the case in which both generalized point potentials are described by matrices belonging to the lower triangular subgroup of $SL(2,\R)$. This subgroup is particularly interesting since the common point potentials, namely the $\delta$, $\delta'$, and $\delta$-$\delta'$ belong to this subgroup under the interpretation of \cite[Sec. 4]{ADK98} for $\delta'$. To keep the analysis simple, we assume Dirichlet boundary conditions at the endpoints of the interval $I$, that is, $\alpha=\beta=0$. 
The matrices $M_i$, $i\in\{1,2\}$, can be parametrized as 
\begin{align}\label{174}
    M_i=\begin{pmatrix}
        a_i& 0\\
        c_i& a_{i}^{-1}
    \end{pmatrix},\quad a_i\in\R/\{0\},\quad c_i\in\R. 
\end{align}
With the parametrization \eqref{174}, the matrices $N$, $Q$, and $P$ become
\begin{align}
    N=\begin{pmatrix}
        a_1a_2& 0\\
        c_2a_1+c_1a_{2}^{-1}&(a_1a_2)^{-1}
    \end{pmatrix},\quad Q=\begin{pmatrix}
        a_2c_1& a_2a_1^{-1}\\
        c_1c_2& c_2a_{1}^{-1}
    \end{pmatrix},\quad P=\begin{pmatrix}
        0&0\\
        a_1a_2^{-1}& 0
    \end{pmatrix}.
\end{align}
By using these matrices in equation \eqref{165} we find that 
\begin{align}
    \sigma=\frac{1}{2}\quad\textrm{and}\quad \rho_{1/2}=-i(a_1+a_{1}^{-1})(a_2+a_{2}^{-1}).
\end{align}
Assuming $m_0=0$ once again, the value in \eqref{172} reduces, in this particular example, to
\begin{align}\label{177}
   F_{0,0,[M]}(0)&=-
R_1(L-R_2)(c_2a_1+c_1a_2^{-1}+Rc_1c_2)
-a_{1}^{-1}(L-R_2)(a_2^{-1}+Rc_2)\nonumber\\
&-a_2R_1(a_1+Rc_1)-R\,a_2a_1^{-1}. 
\end{align}

Thus we conclude that the case of generalized point potentials described by lower triangular matrices with Dirichlet boundary at the endpoints yields (where the last quantity can be evaluated from \eqref{177})
\begin{align}\label{178}
    \zeta'(0;T_{0,0,[M]})=\ln\left[\frac{(a_1+a_{1}^{-1})(a_2+a_{2}^{-1})}{8}\right]-\ln F_{0,0,[M]}(0).
\end{align} 

\section{Concluding remarks} 

In this work, we have developed a framework for the  analysis of the spectral $\zeta$-function associated with Sturm--Liouville operators in the presence of general multiple point potentials. The generalized point interactions are characterized by self-adjoint extensions of the Sturm--Liouville operator defined on a chain of $N+1$ adjacent intervals, with each interaction described by
matching conditions at the common boundary of two neighboring intervals. This formulation allows for generalized point potentials described by matrices in $SL(2,\mathbb{R})$ and provides a unified treatment of systems containing an arbitrary number of point interactions. 

In order to construct the spectral $\zeta$-function of the general $N$ point potential system, we have found the characteristic function, whose zeros determine the spectrum of the
corresponding self-adjoint realization, for the three possible configurations of the endpoints of the interval. The characteristic function includes the influence of the generalized point potentials through its dependence of the transfer matrix. 

We have used a contour integral involving the characteristic function to represent the spectral $\zeta$-function in a region of the complex plane to the right of its abscissa of convergence (see Theorem \ref{Thrm3.4}). The analytic continuation of the spectral $\zeta$-function to a larger region of the complex plane can be achieved by utilizing the large-$z$ asymptotic expansion of the characteristic function. Although the characteristic function can be constructed for an arbitrary number $N$ of point potentials, its general asymptotic analysis becomes increasingly involved as $N$ increases because of the structure of the associated transfer  matrix. For the sake of simplicity, we therefore carried out explicitly the analytic continuation for two representative systems to a neighborhood of $s=0$: An interface problem involving a single point potential separating two different Sturm--Liouville operators, and two generalized point potentials with the same Sturm--Liouville operator on the three resulting intervals.  
For both of these systems we have explicitly computed the value $\zeta'(0;T)$ for self-adjoint realizations $T$ of these problems, which can be used to evaluate their functional determinant.

\end{document}